%% file: main.tex
\documentclass[11pt,letterpaper,onecolumn]{article}

\usepackage[
    top=1in,
    bottom=1in,
    left=1in,
    right=1in
]{geometry}

\usepackage[T1]{fontenc}
\usepackage{lmodern}
\usepackage{microtype}
\usepackage[hyphens]{url}
\usepackage[hidelinks]{hyperref}
\usepackage[numbers,sort&compress]{natbib}

\usepackage{amsmath}
\usepackage{amssymb}
\usepackage{amsthm}
\usepackage{bm}
\usepackage{algorithm}
\usepackage{algpseudocode}

\usepackage{graphicx}
\usepackage{caption}
\usepackage{subcaption}
\usepackage{booktabs}
\usepackage{colortbl}
\usepackage{multirow}
\usepackage{enumitem}
\usepackage{adjustbox}
\usepackage{placeins}

\usepackage{tikz}
\usetikzlibrary{spy,calc}

\theoremstyle{plain}
\newtheorem{theorem}{Theorem}[section]
\newtheorem{lemma}[theorem]{Lemma}
\newtheorem{corollary}[theorem]{Corollary}
\newtheorem{proposition}[theorem]{Proposition}

\theoremstyle{definition}

\newtheorem{assumption}[theorem]{Assumption}

\theoremstyle{remark}
\newtheorem{remark}[theorem]{Remark}

\newcommand{\figw}{0.31\textwidth}
\newcommand{\colgap}{0pt}
\newcommand{\rowgap}{0.4em}

\newcommand{\methodfig}[2]{%
    \begin{subfigure}[t]{\figw}
        \centering
        {\scriptsize
        Watermarked/$_{\mathrm{#1}}$
        \hspace{0.8em}
        Attacked/$_{\mathrm{#1}}$\par}
        \vspace{1pt}
        \begin{tikzpicture}[baseline]
            \node[anchor=south west,inner sep=0] (img) at (0,0)
            {\includegraphics[
                width=\linewidth,
                trim=6 6 6 6,
                clip
            ]{#2}};

            \draw[dashed,white,line width=0.6pt]
                ($(img.south west)!0.5!(img.south east)$) --
                ($(img.north west)!0.5!(img.north east)$);
        \end{tikzpicture}
    \end{subfigure}%
}

\newcommand{\zoomfig}[5]{%
    \begin{tikzpicture}[
        spy using outlines={lens={scale=3},
            size=#5,
            connect spies,
            every spy on node/.append style={
                draw=red,
                dashed,
                line width=0.8pt
            },
            every spy in node/.append style={
                draw=red,
                dashed,
                line width=0.8pt
            }
        }
    ]
        \node[inner sep=0] (img)
        {\includegraphics[width=#4]{#1}};

        \spy on
        ($(img.south west)!#2!(img.south east)!#3!(img.north west)$)
        in node[anchor=north west]
        at ($(img.north east)+(-0.05cm,-0.05cm)$);
    \end{tikzpicture}%
}

\title{\textbf{DRIFT: Removing Diffusion Watermarks by Deflecting the Generative Trajectory}}

\author{
    \begin{tabular}{ccc}
        Rui Bao$^{1*}$ & Zheng Gao$^{1*}$ & Xiaoyu Li$^{1}$ \\
        Xiaoyan Feng$^{2}$ & Yang Song$^{1}$ & Jiaojiao Jiang$^{1}$
    \end{tabular}\\[0.75em]
    $^{1}$University of New South Wales\\
    $^{2}$Griffith University
}

\date{}

\begin{document}

\maketitle

\input{abstract}
\input{intro}
\input{related_work}
\input{method}
\input{theory}
\input{exp_setup}
\input{results}

\FloatBarrier

\input{conclusion}

\bibliographystyle{plainnat}
\bibliography{ref}

% ---- Technical appendix ----
\clearpage
\appendix

\input{appendix}

\end{document}

%% file: abstract.tex
\begin{abstract}
Diffusion watermarking embeds verifiable signals into the generative process and commonly verifies them by recovering trajectory-dependent evidence, making the marks robust to conventional pixel-space distortions. Existing removal attacks either regenerate along deterministic trajectories, which often preserve the watermark-bearing latent structure, or optimize every image separately. We identify the reliance on a recoverable generative trajectory as a common attack surface among the schemes we study. Based on this observation, we propose \textbf{DRIFT}, a black-box attack that combines partial forward diffusion with stochastic reverse resampling. Forward re-noising limits source information available to a fixed-depth recovery pipeline, while stochastic reversal supplies alternative noise-driven paths whose removal benefit we isolate through matched sampler comparisons. \emph{Adaptive DRIFT} searches a selected ladder for each image's first verifier-rejected rung and refines fidelity while retaining only updates rejected by the same verifier. At fixed depth, we derive information-theoretic and Wasserstein source-dependence bounds; under realized-ladder monotonicity, the first rejected rung is least distorted among rejected rungs on that ladder, and verifier-gated refinement preserves rejection. Across nine watermarks spanning three paradigms, DRIFT achieves $98$--$100\%$ attack success and the best image quality among the compared attacks, without secret keys, verifier internals, or per-image gradient optimization.
\end{abstract}

%% file: intro.tex
\begingroup
\setlength{\parskip}{0pt}

\section{Introduction}

Diffusion models~\citep{ho2020denoising,song2021sde,rombach2022high} have made photorealistic image synthesis widely accessible, intensifying the need to trace AI-generated content. Diffusion watermarking provides an active provenance signal by embedding a mark into generation itself, allowing attribution to persist when metadata is removed or overwritten.

Existing methods inject watermarks at different points: Tree-Ring~\citep{Tree-Ring}, RingID~\citep{ci2024ringid}, and PRC~\citep{gunn2024undetectable} structure the initial noise; Gaussian Shading~\citep{yang2024gaussian}, SFW~\citep{lee2025semantic}, and SEAL~\citep{arabi2025seal} constrain latent representations or bind the mark to content; ROBIN~\citep{huang2024robin} learns hidden prompts. Despite these differences, the evaluated public implementations recover evidence associated with their generation or inversion trajectories, typically through deterministic inversion. This dependence is associated with robustness to JPEG compression, blur, crop, and rotation, but may also expose a common attack surface.

Prior attacks do not fully exploit it. Deterministic regeneration can remove pixel-space marks~\citep{zhao2024invisible}, yet often preserves the latent structure used by diffusion watermarks. Latent-removal and black-box attacks~\citep{jain2025forging,muller2025black} can evade stronger schemes, but require iterative optimization for every image. These approaches target individual watermark signals; they leave the shared verification mechanism largely untouched.

Across these implementations, successful verification is empirically associated with recoverable trajectory-linked evidence. We call this reliance \emph{trajectory consistency}. Rather than perturbing a detector-specific signal, an attacker can target the path itself. Based on this principle, we propose \textbf{DRIFT}, which first applies partial forward diffusion to limit source information entering a fixed-depth recovery pipeline and then performs stochastic reverse resampling. Fresh reverse noise opens alternative reconstruction paths, while the pretrained score function promotes natural-image outputs. Whether this variation survives the downstream map and crosses a verifier boundary is scheme- and sampler-dependent; we establish its removal benefit empirically through matched sampler comparisons.

A fixed re-noising strength wastes fidelity because watermark robustness varies across schemes and images. \emph{Adaptive DRIFT} therefore predicts the watermark family from a shared inverted latent, uses binary verifier feedback to stop at the first verifier-rejected rung on the selected ladder, and applies a DPPO-trained refinement controller~\citep{ren2024dppo} with verifier-gated back-off to recover quality.

Our contributions are: (i)~identifying trajectory consistency across nine schemes; (ii)~introducing DRIFT with fixed-depth information and Wasserstein bounds, conditional ladder-relative selection, and verifier-rejection invariance; and (iii)~isolating stochasticity with matched samplers and obtaining $98$--$100\%$ success with the best compared fidelity.

\par
\endgroup

%% file: related_work.tex
\section{Related Work}
\label{sec:related_work}

\paragraph{Diffusion sampling and inversion.}
Deterministic samplers such as DDIM~\citep{song2021ddim} and higher-order
variants~\citep{lu2022dpmsolver,zhao2023unipc,liu2022pndm} define a coupled
noise-to-image path, the property underpinning DDIM
inversion~\citep{mokady2023nulltext,wallace2023edict} and most diffusion
watermark verifiers, though prediction errors accumulate along
it~\citep{lin2024schedule,blasingame2025reversible}.  Reverse-time SDE
samplers~\citep{song2021sde,xue2023sasolver,gonzalez2023seeds} instead inject
fresh noise during generation, so repeated runs from the same state can follow
different trajectories~\citep{nie2024sdebeatsode}; we exploit this contrast.

\paragraph{Diffusion watermarking.}
Watermarks differ by injection point.  \emph{Noise-space} methods modify the
initial latent: Tree-Ring~\citep{Tree-Ring} writes a Fourier ring pattern,
RingID~\citep{ci2024ringid} extends it to multi-channel patterns,
PRC~\citep{gunn2024undetectable} samples pseudorandom codes, and
WIND~\citep{arabi2024hidden} organizes large key pools.  \emph{Latent- and
frequency-domain} methods act on intermediate representations: Gaussian
Shading~\citep{yang2024gaussian} applies key-controlled spectral offsets,
GaussMarker~\citep{li2025gaussmarker} encodes high-frequency components, and
SFW~\citep{lee2025semantic} and SEAL~\citep{arabi2025seal} bind marks to
semantic content; optimization-based ROBIN~\citep{huang2024robin} learns hidden
prompts.  Although these embeddings differ, the evaluated implementations
recover structure tied to the generation or inversion trajectory.

\paragraph{Removal attacks.}
Regeneration attacks reconstruct the image with a diffusion model:
\citet{zhao2024invisible} use deterministic PNDM with guarantees for
pixel-level marks, CtrlRegen~\citep{liu2024image} adds trained control modules,
\citet{saberi2023robustness} adapt DiffPure through a reverse-time SDE, and
DDWRM~\citep{mareen2024diffusion} denoises in pixel space; others optimize a
latent perturbation per input~\citep{jain2025forging,muller2025black}.  DRIFT
instead isolates stochastic trajectory deflection under matched samplers,
requires no per-image gradients, and draws on inversion-based fingerprints and
diffusion-policy optimization~\citep{teng2025fingerprinting,black2023ddpo,
ren2024dppo,hu2025better} only to seed a verifier-guided strength search and
to retain refinement updates when rejection is preserved.

%% file: method.tex
\section{The DRIFT Attack}
\label{sec:method}

\begin{figure*}[t]
    \centering
    \includegraphics[width=0.88\textwidth]{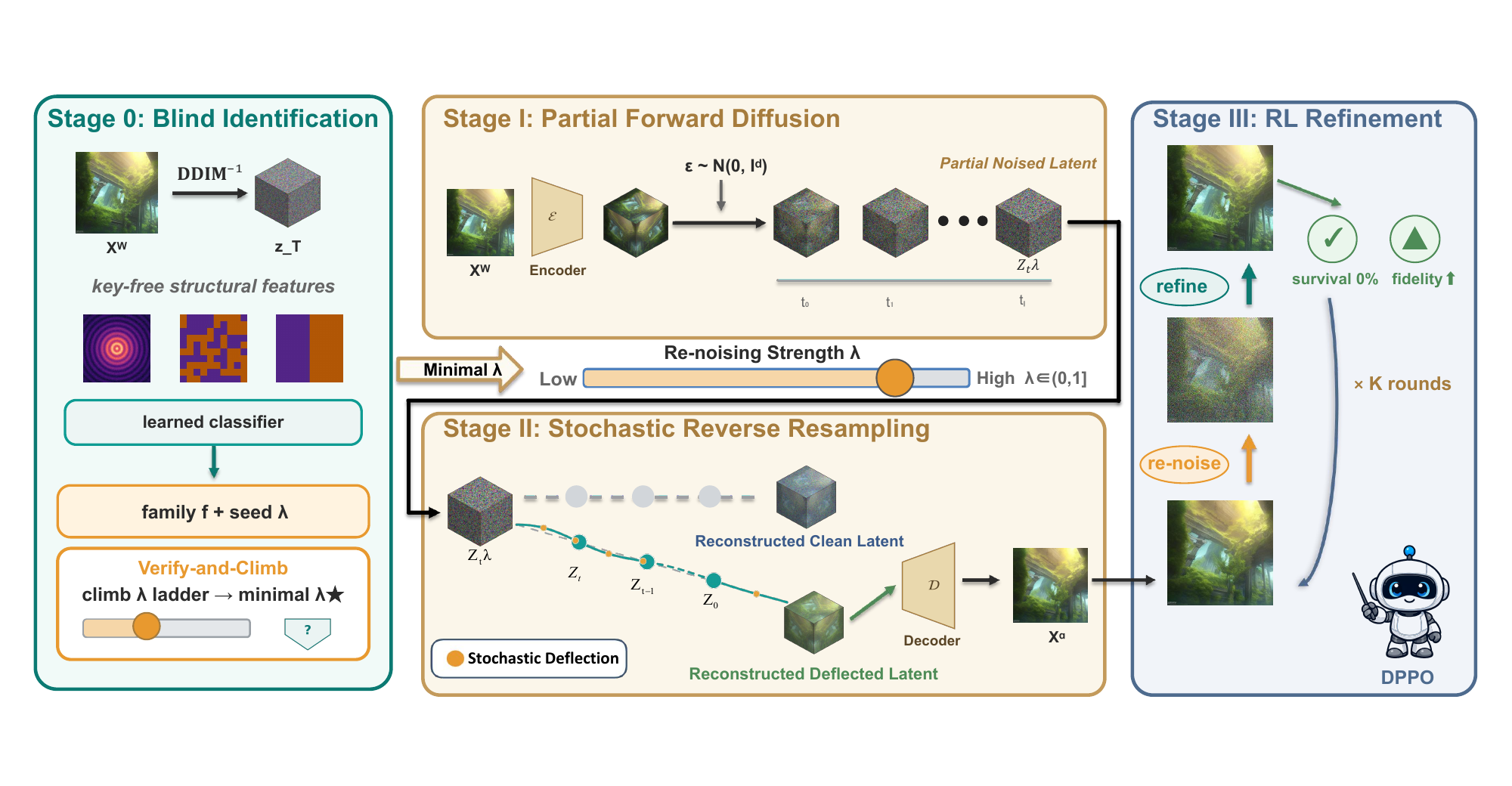}
    \caption{\textbf{Overview of Adaptive DRIFT.} A shared inversion predicts the watermark family and seeds an ascending strength ladder. DRIFT then combines partial forward diffusion with stochastic reverse resampling; an optional controller recovers fidelity while retaining only candidates rejected by the same verifier.}
    \label{fig:framework}
\end{figure*}

\noindent\textbf{Problem setup.}
Let $\mathcal{W}_f$ be a diffusion-watermarking scheme from family $f$.
Given prompt $c$ and secret key $\kappa$, it produces
$\mathbf{x}^w=\mathcal{W}_f(c,\kappa)$.  Its verifier computes
$s_{f,\kappa}(\mathbf{x})$ and returns
$\mathcal{V}_{f,\kappa}(\mathbf{x})
=\mathbf{1}[s_{f,\kappa}(\mathbf{x})\ge\tau]$, covering both zero-bit and
message-based schemes.  For a target
$\delta_{\mathrm{fail}}\in[0,1]$, a randomized attack
$\mathcal{A}_{\boldsymbol{\omega}}$ returns
$\mathbf x_{\boldsymbol{\omega}}^a
=\mathcal{A}_{\boldsymbol{\omega}}(\mathbf x^w)$ and seeks low perceptual
distortion while meeting a target failure probability:
\begin{equation}
    \begin{aligned}
    \min_{\mathcal A}\quad&
    \mathbb E_{\boldsymbol{\omega}}
    [d_{\mathrm{sem}}(\mathbf x_{\boldsymbol{\omega}}^a,\mathbf x^w)]
    \\
    \mathrm{s.t.}\quad&
    \Pr_{\boldsymbol{\omega}}\!\left[
    \mathcal V_{f,\kappa}(\mathbf x_{\boldsymbol{\omega}}^a)=0\right]
    \ge 1-\delta_{\mathrm{fail}} .
    \end{aligned}
    \label{eq:attack_objective}
\end{equation}
The attacker observes only $\mathbf{x}^w$ and uses a public latent diffusion
model $(\mathcal E,\mathcal D,\boldsymbol{\epsilon}_\theta)$, which need not
match the watermarked generator.  The key, message, prompt, family, and
verifier internals are unknown.  All denoising and inversion operations use
the same globally fixed deterministic attacker-side conditioning
$c_{\mathrm{att}}$ (e.g., the empty prompt), independent of the source image
and all attack randomness; we suppress it in the notation below.  Base DRIFT
is training-free and query-free; the
adaptive stages receive only binary verifier responses and never perform
per-image gradient optimization.

\noindent\textbf{Attack overview.}
The evaluated watermark families use different signals, but their verifiers
all recover evidence coupled to the marked generation or inversion path.
DRIFT targets this shared interface.  Partial forward diffusion attenuates
the explicit source-latent contribution, while stochastic reverse sampling
reconstructs through fresh noise-driven transitions.  Adaptive DRIFT then
uses a blind family prediction to seed a finite strength ladder, returns its
first verifier-rejected candidate, and optionally refines that candidate
under a hard verifier gate.  Figure~\ref{fig:framework} summarizes the three
stages.

\subsection{Stage I: Partial Forward Diffusion}

For a diffusion horizon $T\in\mathbb N$, we encode the watermarked image as
$\mathbf{z}_0^w=\mathcal E(\mathbf{x}^w)$ and map strength
$\lambda\in(0,1]$ to $t_\lambda=\lceil\lambda T\rceil$.  The closed-form
DDPM forward process~\citep{ho2020denoising} gives
\begin{equation}
    \mathbf{z}_{t_\lambda}^a
    =\sqrt{\bar\alpha_{t_\lambda}}\,\mathbf{z}_0^w
    +\sqrt{1-\bar\alpha_{t_\lambda}}\,\boldsymbol{\epsilon},
    \qquad
    \boldsymbol{\epsilon}\sim\mathcal N(\mathbf 0,\mathbf I_d),
    \label{eq:forward}
\end{equation}
where $\bar\alpha_t=\prod_{i=1}^t(1-\beta_i)$ and
$\bar\alpha_0=1$.  Increasing $\lambda$ decreases the explicit coefficient
$\sqrt{\bar\alpha_{t_\lambda}}$ of the source latent, but also discards more
source detail.  This is the removal--fidelity trade-off addressed by the
adaptive ladder.  Section~\ref{sec:theory} formalizes the corresponding
forward information bottleneck without assuming that the reverse network is
contractive.

\subsection{Stage II: Stochastic Reverse Resampling}

Starting from $\mathbf{z}_{t_\lambda}^a$, the denoiser first predicts
\begin{equation}
    \widehat{\mathbf z}_0^a
    =
    \frac{\mathbf z_t^a-\sqrt{1-\bar\alpha_t}\,
    \boldsymbol{\epsilon}_\theta(\mathbf z_t^a,t)}
    {\sqrt{\bar\alpha_t}} .
    \label{eq:predict_clean}
\end{equation}
DRIFT then uses the generalized DDIM/DDPM reverse update
\begin{align}
    \mathbf z_{t-1}^a
    &=
    \sqrt{\bar\alpha_{t-1}}\,\widehat{\mathbf z}_0^a
    +\sqrt{1-\bar\alpha_{t-1}-\sigma_t^2}\,
      \boldsymbol{\epsilon}_\theta(\mathbf z_t^a,t)
      \notag\\
    &\quad+
    \underbrace{\sigma_t\boldsymbol{\xi}_t}
    _{\text{stochastic deflection}},
    \qquad
    \boldsymbol{\xi}_t\sim\mathcal N(\mathbf 0,\mathbf I_d).
    \label{eq:reverse}
\end{align}
Here
$0\le\sigma_t\le\sqrt{1-\bar\alpha_{t-1}}$, so $\sigma_1=0$.
Each $\boldsymbol{\xi}_t$ is fresh relative to the current reverse history.
The choice $\sigma_t=0$ recovers deterministic DDIM, whereas the posterior
variance gives the standard DDPM ancestral special case; other admissible
scales define generalized stochastic DDIM samplers.  A positive $\sigma_t$
makes that transition non-degenerate.  Terminal and image-level diversity
additionally require that subsequent reverse steps and the decoder do not
collapse the injected variation.  Thus stochasticity supplies path
deflection, while its removal benefit is established by the controlled
sampler comparison in Section~\ref{ssec:det_vs_stoch}, not assumed as a
universal theorem.

The final latent is decoded as
$\mathbf{x}^a=\mathcal D(\mathbf z_0^a)$.  Base DRIFT uses one VAE encode,
$t_\lambda$ denoiser evaluations, and one VAE decode, matching the order of a
standard image-to-image diffusion pass.  Algorithm~\ref{alg:drift} appears in
Appendix~\ref{app:algorithms}.

\subsection{Adaptive Strength Selection}
\label{ssec:adaptive}

A global strength either fails on robust marks or unnecessarily degrades easy
instances.  Adaptive DRIFT therefore computes one shared inversion
\[
    \boldsymbol{\epsilon}^w
    \equiv\mathbf z_T^{\mathrm{inv}}
    =\mathcal I_{\mathrm{DDIM}}(\mathbf z_0^w)
\]
and extracts radial Fourier power, sign-field tiling, cross-channel
correlation, and short-period block statistics.  A lightweight classifier
$\widehat f=\mathcal H(\mathbf z_T^{\mathrm{inv}})$ predicts the family
without observing the key or message, repurposing inversion fingerprints
used for source attribution~\citep{teng2025fingerprinting}.  The identifier
is an empirical query-saving device; the search guarantee below applies to
whichever ladder it selects and does not require $\widehat f=f$.

The prediction chooses an offline seed and an ascending ladder.  We fix
$\delta>0$, $n_{\downarrow}\in\mathbb N_0$, and
$0<\lambda_{\mathrm{seed}}(g)\le\lambda_{\max}\le1$, and define
$J_g=\lfloor(\lambda_{\max}-\lambda_{\mathrm{seed}}(g))/\delta\rfloor$.
Then
\begin{equation}
    \begin{aligned}
    \Lambda_{\widehat f}
    =\operatorname{sort}\!\bigl\{
    &\lambda_{\mathrm{seed}}(\widehat f)+j\delta:
    j=-n_{\downarrow},\ldots,J_{\widehat f},\\[-2pt]
    &0<\lambda_{\mathrm{seed}}(\widehat f)+j\delta
    \le\lambda_{\max}\bigr\},
    \end{aligned}
    \label{eq:lambda_ladder}
\end{equation}
The parameter conditions guarantee that the finite ladder is nonempty.
DRIFT queries it from its lowest retained rung and stops at the first verifier
rejection.  If no rung is rejected, it returns the strongest evaluated
candidate with a failure flag; this outcome counts as an unsuccessful attack
in Eq.~\eqref{eq:attack_objective}.  If verdicts form a rejected suffix and
distortion is non-decreasing along the realized ladder, the first rejected
candidate is the least distorted rejected rung on that ladder
(Proposition~\ref{prop:minimal}).  It is not claimed to be the optimum
between grid points or over the continuous objective.  The complete
predict--seed--climb procedure is Algorithm~\ref{alg:adaptive}.

The distributional bounds in Section~\ref{sec:theory} concern a strength fixed
independently of the source and attack randomness.  They therefore apply to
base DRIFT or a prespecified ladder rung, not automatically to the final
candidate selected using the identifier and verifier history.

\subsection{Feasibility-Preserving Fidelity Refinement}
\label{ssec:rl}

The first rejected rung may still lose fine detail.  Starting from that image
$\mathbf x_0$, Stage~III performs short guided-SDEdit
rounds~\citep{meng2022sdedit,chung2023dps}.  At round $k$, the retained image
is lightly re-noised by $\eta_k$ and denoised with pixel-MSE and LPIPS
guidance~\citep{zhang2018unreasonable} toward $\mathbf x^w$.  The same
globally fixed scalarization is used for retention throughout:
\begin{equation}
    \begin{aligned}
    S(\mathbf x)
    &=
    w_{\mathrm s}\operatorname{SSIM}(\mathbf x,\mathbf x^w)\\[-2pt]
    &\quad
    -w_{\ell}\operatorname{LPIPS}(\mathbf x,\mathbf x^w),
    \qquad w_{\mathrm s},w_{\ell}>0.
    \end{aligned}
    \label{eq:fidelity_score}
\end{equation}
A candidate replaces the retained best only if the same deterministic
verifier rejects it and $S$ strictly improves; otherwise the algorithm backs
off to the previous best.  This update rule, rather than the learned policy,
preserves verifier-relative feasibility.

We learn the refinement schedule as a Markov decision process.  The state
contains DINOv2 embeddings~\citep{oquab2023dinov2} of the source and retained
images, the action/fidelity history, and the round index, but not the family
label.  The action
\[
    \mathbf a_k
    =(\eta_k,g_{\mathrm{pix},k},g_{\mathrm{lpips},k},\mathrm{stop})
\]
controls re-noising, the two guidance weights, and termination.  With metric
increments defined as current minus previous, the step reward is
\[
    r_k^{\mathrm{step}}
    =
    \Delta S_k-w_c\,\mathrm{steps}_k,
\]
followed by terminal rejection and fidelity bonuses.  Hence increasing SSIM
or decreasing LPIPS raises reward.

The controller is a diffusion policy trained with
DPPO~\citep{ren2024dppo}: a short conditional denoising chain proposes each
action, and PPO optimizes the summed Gaussian log-probability with clipping
and generalized advantage estimation~\citep{schulman2017ppo,schulman2016gae}.
Reward shaping and back-off follow B2-DiffuRL~\citep{hu2025better}.  For
any fixed rollout, returning the retained best keeps it rejected by the same
verifier and makes $S$ non-decreasing along nested prefixes
(Propositions~\ref{prop:feasible}--\ref{prop:monotone_fid}); these invariants
do not assert monotone PPO return or improvement of every component metric.

%% file: theory.tex
\section{Analysis of DRIFT}
\label{sec:theory}

We separate three questions that are easy to conflate: what forward
re-noising removes about the source, what stochastic reversal changes, and
what the adaptive loops guarantee.  The first admits a distributional
theorem; the second is a structural distinction whose removal benefit is
empirical; the third follows from finite search and verifier-gated back-off.
Appendix~\ref{app:theory} gives complete proofs, including the fully
quantified fixed-depth result in Theorem~\ref{prop:noise_distance_formal} and
the conditional terminal diagnostic in Corollary~\ref{cor:shift_2d}.

\subsection{Forward Re-noising as an Information Bottleneck}
\label{ssec:bound}

The attack in Section~\ref{sec:method} acts on a fixed image.  For the
distributional statements below, however, we draw a watermarked image from
the source population induced by the generation protocol (including prompts,
keys, messages, and generator randomness), independently of the attack
randomness.  Mutual information, Wasserstein distance, and unconditioned
expectations are taken over this population and the fresh attack noises; the
fixed-image sensitivity statement below conditions on the source latent.

Write $\mathbf X=\mathbf z_0^w$ for the source latent,
$\mathbf W=\boldsymbol{\epsilon}^w$ for its recovered reference noise, and
\[
    \mathbf Z_n
    =\sqrt{\bar\alpha_n}\mathbf X
    +\sqrt{1-\bar\alpha_n}\boldsymbol{\epsilon}
\]
for the Stage-I state at depth $n$.  The complete Stage-II and recovered-noise
pipeline is a measurable randomized map
$\mathbf Y_n=H_n(\mathbf Z_n,\boldsymbol{\Xi}_{1:n})
=\widehat{\boldsymbol{\epsilon}}_n^a$.  To expose the source contribution,
we also run the same reverse-noise realization from
$\widetilde{\mathbf Z}_n
=\sqrt{1-\bar\alpha_n}\boldsymbol{\epsilon}$, yielding the source-free
baseline $\widetilde{\boldsymbol{\epsilon}}_n$.  Let
$\Sigma_z=\operatorname{Cov}(\mathbf z_0^w)$.

\begin{theorem}[Fixed-depth source-dependence bounds]
\label{prop:shift_informal}
Fix a deterministic $\lambda\in(0,1]$ before sampling the source and attack
randomness, and set $n=t_\lambda$.  Under the
source-independent-conditioning, measurability,
independence, and moment conditions in
Assumption~\ref{asm:shift_regular}, the following hold.
\begin{enumerate}[label=(\roman*),nosep,leftmargin=*]
    \item The forward channel imposes
    \begin{align}
        I(\boldsymbol{\epsilon}^w;
          \widehat{\boldsymbol{\epsilon}}_n^a)
        &\le
        \frac12\log\det\!\left(
        \mathbf I_d+
        \frac{\bar\alpha_n}{1-\bar\alpha_n}\Sigma_z
        \right)
        \notag\\
        &\le
        \frac d2\log\!\left(
        1+
        \frac{\bar\alpha_n\operatorname{tr}(\Sigma_z)}
        {d(1-\bar\alpha_n)}
        \right)
        =:\mathcal B_n ,
        \label{eq:information_bound}
    \end{align}
    and the total-variation distance between the joint law and the product of
    its marginals is at most $\sqrt{\mathcal B_n/2}$.
    \item If the reverse pipeline is $L_QC_n$-Lipschitz under synchronous
    reverse-noise coupling, with
    $\Gamma_n=\sqrt{\bar\alpha_n}C_n$, then for each fixed source latent
    $\mathbf x$ the coupled difference is pathwise at most
    $L_Q\Gamma_n\|\mathbf x\|_2$.  Averaging over the source population gives
    \begin{equation}
        \mathbb E\!\left[
        \|\widehat{\boldsymbol{\epsilon}}_n^a
        -\widetilde{\boldsymbol{\epsilon}}_n\|_2^2
        \right]
        \le
        L_Q^2\Gamma_n^2
        \mathbb E\|\mathbf z_0^w\|_2^2
        =:\Delta_n .
        \label{eq:source_sensitivity}
    \end{equation}
    With the stated second-moment conditions,
    \[
        W_2\!\left(
        \mathcal L(\widehat{\boldsymbol{\epsilon}}_n^a,
                   \boldsymbol{\epsilon}^w),
        \mathcal L(\widehat{\boldsymbol{\epsilon}}_n^a)
        \otimes\mathcal L(\boldsymbol{\epsilon}^w)
        \right)
        \le 2\sqrt{\Delta_n}.
    \]
    \item Separately, when $\lambda=1$ and $n=T$, if
    Assumption~\ref{asm:shift_regular} is instantiated at $T$ and the
    additional terminal premise in Assumption~\ref{asm:shift_prior} holds,
    \[
        \left|
        \mathbb E\|
        \widehat{\boldsymbol{\epsilon}}_T^a
        -\boldsymbol{\epsilon}^w\|_2^2-2d
        \right|
        \le
        \Delta_T+2\sqrt{2d\,\Delta_T}.
    \]
\end{enumerate}
\end{theorem}

The theorem applies to base DRIFT or a prespecified ladder rung.  It does not
automatically apply to the final Adaptive DRIFT output: its selected depth
depends on the identifier and verifier history and can itself carry source
information.  The ladder and refinement guarantees in
Section~\ref{ssec:why_minimal} are separate pathwise statements.

\noindent\textbf{Proof idea.}
Fresh attack randomness gives the Markov chains
$\boldsymbol{\epsilon}^w\!\to\!\mathbf z_0^w\!\to\!\mathbf Z_n
\!\to\!\widehat{\boldsymbol{\epsilon}}_n^a$; data processing and the
Gaussian-channel capacity bound yield part~(i).  For part~(ii), couple the
attacked and source-free initializations with identical forward and reverse
noises: their only initial difference is
$\sqrt{\bar\alpha_n}\mathbf z_0^w$, which the reverse pipeline amplifies by
at most $L_QC_n$.  A second, tensorized coupling converts this mean-square
estimate into the product-law Wasserstein bound.  Part~(iii) expands the
squared distance around the independent terminal reference and controls the
cross term by Cauchy--Schwarz.

\noindent\textbf{Why this is true.}
In one dimension, Stage~I is simply a Gaussian channel whose
signal-to-noise ratio is
$\bar\alpha_n/(1-\bar\alpha_n)$; no downstream randomized map can recover
more information about the source than entered that channel.  The coupling
view gives the complementary geometric statement: removing the source term
changes the initialization by exactly
$\sqrt{\bar\alpha_n}z_0^w$, after which only the sensitivity of the shared
reverse map matters.

\noindent\textbf{Takeaway.}
Forward re-noising places a monotone information envelope on any recovered
signal whose only access to the source is through the fixed-depth state
$\mathbf Z_n$, while the complementary $W_2$ bound quantifies proximity to a
source-free baseline when the realized reverse map is sufficiently
insensitive.

Only $\mathcal B_n$ and the forward-channel information
$I(\boldsymbol{\epsilon}^w;\mathbf Z_n)$ are guaranteed to decrease with
depth.  The actual post-processed information,
$\Gamma_n$, $\Delta_n$, and their $W_2$ envelope need not be monotone because
the reverse map changes with $n$.  Nor does any dependence metric reveal an
unknown verifier margin.  The paired latent distances and strength thresholds
in Section~\ref{ssec:latent_evidence} and
Figure~\ref{fig:lambda_sensitivity} of
Appendix~\ref{app:lambda_sensitivity} are distinct empirical diagnostics, not
estimates of $\mathcal B_n$ or of the joint/product-law $W_2$ above.  The
terminal value $2d$ is a conditional diagnostic, not a distributional claim
for every watermark family.

\subsection{What Stochastic Resampling Changes}
\label{ssec:why_stochastic}

\begin{remark}[Path randomness is local; removal advantage is empirical]
\label{rem:stochasticity}
Conditioned on $\mathbf z_{t_\lambda}^a$, deterministic DDIM returns a point
mass.  At a step with $\sigma_t>0$, Eq.~\eqref{eq:reverse} instead has a
non-Dirac next-state law.  This local fact does not by itself make the final
latent or decoded image non-Dirac: later reverse steps or the decoder may
collapse the injected variation.  Precisely, terminal diversity holds if
and only if two independent complete noise rollouts disagree with positive
probability after the full downstream map
(Lemma~\ref{lem:shift_path_diversity}).
\end{remark}

This distinction also explains why Theorem~\ref{prop:shift_informal} does not
prove a stochastic advantage.  Its sensitivity argument drives two
trajectories with the same reverse noises, which cancel; the same bound holds
when every $\sigma_t=0$.  Stochastic resampling supplies alternative
noise-driven paths, but whether those paths cross a watermark verifier's
decision boundary is scheme- and sampler-dependent.  Section~\ref{ssec:det_vs_stoch}
therefore holds re-noising depth, inputs, backbone, and strength grid fixed
and changes only the reverse-sampler condition.

\subsection{Operational Guarantees of Adaptive DRIFT}
\label{ssec:why_minimal}

\noindent\textbf{Verify-and-climb.}
Fix the family prediction and one complete realization of the candidates on
the selected ladder
$\Lambda_{\widehat f}
=\{\lambda_{\widehat f,0}<\cdots<\lambda_{\widehat f,M}\}$.
Let $v_j$ be its binary verdict and
$\phi_j=d_{\mathrm{sem}}(\mathbf x_j^a,\mathbf x^w)$.

\begin{assumption}[Monotone realized ladder]
\label{asm:monotone}
The verdicts form a rejected suffix,
$v_{j+1}\le v_j$, and distortions are non-decreasing,
$\phi_{j+1}\ge\phi_j$, on this fixed realized ladder.
\end{assumption}

\begin{proposition}[First rejected rung]
\label{prop:minimal}
Suppose at least one rung is rejected and
$j^\star=\min\{j:v_j=0\}$.  Algorithm~\ref{alg:adaptive} returns rung
$j^\star$ after exactly $j^\star+1$ ladder-verifier queries, and
$\phi_{j^\star}\le\phi_j$ for every rejected rung $j$ on the selected
realized ladder.
\end{proposition}

The result is deliberately grid-relative.  It neither selects an optimum
between rungs nor compares independently resampled candidates.  The family
identifier affects which rungs are evaluated, not the validity of the finite
scan: its accuracy and query savings are therefore evaluated empirically.

\noindent\textbf{Verifier-gated refinement.}
Let $b_0$ be the first rejected attack output and let $b_k$ be the retained
best after round $k$.  The following statements use the same fixed
deterministic verifier and the same fixed composite score $S$ throughout.

\begin{proposition}[Feasibility invariance]
\label{prop:feasible}
If $b_0$ is verifier-rejected and $b_k$ is replaced only by a
verifier-rejected candidate, then every retained $b_k$ and the returned
image are rejected by that verifier, independently of the controller.
\end{proposition}

\begin{proposition}[Monotone retained score]
\label{prop:monotone_fid}
Along nested prefixes of one fixed rollout, if a candidate replaces
$b_{k-1}$ only when $S$ strictly improves, then
$S(b_k)\ge S(b_{k-1})$ for every $k$.
\end{proposition}

Both propositions are invariants of the retained-best update, not guarantees
about raw candidates or the learned policy.  In particular, a higher
composite score need not improve every component metric, and PPO clipping
does not imply monotone true return.  Section~\ref{ssec:refine_results}
measures the policy's actual fidelity gain; the propositions certify only
that the same rollout never deploys a gate-violating or lower-scoring
retained image.  Appendices~\ref{app:adaptive} and~\ref{app:rl} give the
proofs and the precise DPPO objective.

%% file: exp_setup.tex
\section{Experiments}
\label{sec:exp_setup}
\label{sec:results}

In this section, we evaluate DRIFT along four questions: whether it removes watermarks across different embedding paradigms, whether stochastic reverse sampling contributes beyond re-noising, whether per-image adaptation reduces unnecessary distortion, and whether the full pipeline preserves visual fidelity. We use Stable Diffusion v2.1 as the public diffusion backbone, draw evaluation prompts from Stable-Diffusion-Prompts~\citep{Santana2024StableDiffusionPrompts}, and run all experiments on an H100 GPU.

\textbf{Evaluated watermarks.} We consider nine representative diffusion watermarks spanning three injection points. The \emph{noise-space} group contains Tree-Ring (TR)~\citep{Tree-Ring}, RingID (RI)~\citep{ci2024ringid}, PRC~\citep{gunn2024undetectable}, and WIND~\citep{arabi2024hidden}; the \emph{latent/frequency-domain} group contains Gaussian Shading (GS)~\citep{yang2024gaussian}, GaussMarker (GM)~\citep{li2025gaussmarker}, SFW~\citep{lee2025semantic}, and SEAL~\citep{arabi2025seal}; and the \emph{optimization-based} group contains ROBIN~\citep{huang2024robin}. We use these abbreviations throughout the tables and appendix.

\textbf{Attack baselines.} We compare against the black-box attack~\citep{muller2025black} and the latent-noise removing attack~\citep{jain2025forging}, two prior attacks that target semantic watermarks through per-image adversarial optimization. DRIFT performs no per-image gradient optimization.

\textbf{Metrics and protocol.} Attack success rate (ASR, $\uparrow$) is the true-detector removal rate. We measure distribution-level quality with CLIP score ($\uparrow$)~\citep{hessel2021clipscore} and FID ($\downarrow$)~\citep{heusel2017gans}, and paired fidelity against the watermarked original with SSIM/PSNR ($\uparrow$) and LPIPS ($\downarrow$). For the adaptive and refinement stages, we additionally report identification accuracy, the selected per-image strength $\lambda$, verifier queries/runtime, and pre$\to$post fidelity. All controlled comparisons use the same encode\,$\to$\,deflect\,$\to$\,decode route: Section~\ref{ssec:det_vs_stoch} holds the re-noising stage, images, backbone, and strength grid fixed while changing the reverse-sampler condition; Section~\ref{ssec:minimal_results} changes only how $\lambda$ is selected; and Section~\ref{ssec:refine_results} changes whether Stage~III is applied. Full per-family results, ablations, and qualitative panels are provided in Appendix~\ref{app:experiments}.

% Numerical entries are preserved verbatim; author-identified corrections are pending.
\begin{table*}[t]
\centering
\caption{\textbf{Overall attack effectiveness and image quality.} \textbf{(a)} Per-watermark ASR under the strong DRIFT setting, $\lambda\geq0.5$; underlined entries indicate where DRIFT matches or exceeds the best prior attack. \textbf{(b)} End-to-end ASR, image quality, and runtime averaged over the nine families; bold marks the best quality values among the three DRIFT configurations. For the two baselines and Full DRIFT, FID, SSIM, and LPIPS follow the $100$-image-per-family quality protocol detailed in Appendix~\ref{app:quality}.}
\label{tab:main_results}
\footnotesize
\setlength{\tabcolsep}{3.5pt}
\renewcommand{\arraystretch}{1.06}
\textbf{(a) Per-watermark attack success rate (ASR, \%).}\\[-2pt]
\begin{tabular*}{\textwidth}{@{\extracolsep{\fill}}lcccccccccc@{}}
\toprule
\multirow{2}{*}{\textbf{Attack}}
& \multicolumn{4}{c}{\textbf{Noise-space}}
& \multicolumn{4}{c}{\textbf{Latent/frequency}}
& \textbf{Optimization}
& \multirow{2}{*}{\textbf{AVG.}} \\
\cmidrule(lr){2-5}\cmidrule(lr){6-9}\cmidrule(lr){10-10}
& TR & RI & PRC & WIND & GS & GM & SFW & SEAL & ROBIN & \\
\midrule
\rowcolor{black!6}
\multicolumn{11}{c}{\itshape Prior attacks} \\
Black-box & 72 & 99 & 100 & 100 & 100 & 100 & 100 & 100 & 100 & 96.8 \\
Removing  & 99 & 95 & 100 & 100 & 4   & 55  & 100 & 100 & 79  & 81.3 \\
\midrule
\rowcolor{black!6}
\multicolumn{11}{c}{\itshape DRIFT (ours)} \\
\textbf{DRIFT} &
98 & \underline{\textbf{100}} & \underline{\textbf{100}} &
\underline{\textbf{100}} & \underline{\textbf{100}} &
\underline{\textbf{100}} & \underline{\textbf{100}} &
\underline{\textbf{100}} & \underline{\textbf{100}} &
\underline{\textbf{99.8}} \\
\bottomrule
\end{tabular*}

\vspace{4pt}
\textbf{(b) End-to-end effectiveness, image quality, and runtime.}\\[-2pt]
\begin{tabular*}{\textwidth}{@{\extracolsep{\fill}}lccccccc@{}}
\toprule
\textbf{Method} & \textbf{ASR (\%)$\uparrow$} & \textbf{CLIP$\uparrow$} &
\textbf{FID$\downarrow$} & \textbf{SSIM$\uparrow$} & \textbf{PSNR$\uparrow$} &
\textbf{LPIPS$\downarrow$} & \textbf{Time (s)} \\
\midrule
\rowcolor{black!6}
\multicolumn{8}{c}{\itshape Baseline attacks} \\
Black-box attack                         & 96.8 & 31.9 & 114.019 & 0.519 & 19.46 & 0.305 & -- \\
Removing attack~\citep{jain2025forging} & 81.3 & 31.4 & 104.693 & 0.712 & 26.35 & 0.338 & -- \\
\midrule
\rowcolor{black!6}
\multicolumn{8}{c}{\itshape DRIFT (ours), end-to-end configurations} \\
Fixed-$\lambda$ DRIFT ($\lambda=0.70$)   & 100.0 & 31.8 & 136.5 & 0.401 & 14.31 & 0.528 & 1.16 \\
Adaptive-$\lambda$ DRIFT (no refine)     & 99.7  & 31.5 & 58.7 & 0.697 & 23.72 & 0.168 & 4.20 \\
Full DRIFT ($+$ DPPO refine)             & 100.0 & \textbf{32.1} & \textbf{55.414} &
\textbf{0.713} & 24.01 & \textbf{0.154} & 8.12 \\
\bottomrule
\end{tabular*}
\end{table*}

%% file: results.tex
\subsection{Overall Attack Effectiveness}
\label{ssec:main_results}

Table~\ref{tab:main_results}(a) tests one trajectory-level attack across nine schemes under the strong setting, $\lambda\geq0.5$. DRIFT reaches $98$--$100\%$ ASR on every scheme and averages \textbf{99.8\%}, versus $96.8\%$ for the black-box attack and $81.3\%$ for the removing attack. The largest gaps occur on Tree-Ring ($98\%$ vs.\ $72\%$ black-box) and on GaussMarker/Gaussian Shading ($100\%$ vs.\ $55\%/4\%$ removing). Its success across all three watermark paradigms establishes breadth; the controlled study below isolates the contribution of stochastic reversal.

\subsection{Role of Stochastic Reverse Sampling}
\label{ssec:det_vs_stoch}

Theorem~\ref{prop:shift_informal} does not establish a stochastic advantage because its bound also holds at $\sigma_t{=}0$ (Section~\ref{ssec:why_stochastic}). We therefore fix the re-noising stage, images, backbone, and strength grid while comparing three deterministic samplers (DDIM, DPM-Solver++, Euler) with three stochastic samplers (DDPM ancestral, Euler-a, SDE-DPM-Solver++). Each uses three verifier-inversion settings on nine families with $100$ images per cell and $\lambda\in[0,0.70]$. Table~\ref{tab:ode_vs_sde} reports the smallest strength $\lambda^\star_p$ reaching $p\%$ survival ($\downarrow$ better).  ``never'' means the threshold is not reached by $\lambda=0.70$. Accordingly, the $\lambda^\star_5$ mean uses the eight families resolved by both groups, while $\lambda^\star_{50}$ uses all nine shown.

\begin{table}[t]
\centering
\caption{\textbf{Deterministic (ODE) versus stochastic (SDE) reverse sampling.} We report $\lambda^\star_p$, the smallest strength driving watermark survival to $p\%$, averaged over three samplers and three inversion settings ($\downarrow$ better).}
\label{tab:ode_vs_sde}
\setlength{\tabcolsep}{4pt}
\renewcommand{\arraystretch}{1.10}
\small
\begin{tabular}{lcc@{\hspace{8pt}}cc}
\toprule
& \multicolumn{2}{c}{$\lambda^\star_{50}\downarrow$}
& \multicolumn{2}{c}{$\lambda^\star_{5}\downarrow$} \\
\cmidrule(lr){2-3}\cmidrule(lr){4-5}
Family & ODE & SDE & ODE & SDE \\
\midrule
SEAL              & 0.038 & \textbf{0.032} & 0.097 & \textbf{0.090} \\
SFW              & 0.040 & \textbf{0.038} & 0.093 & \textbf{0.090} \\
WIND             & 0.100 & \textbf{0.076} & 0.197 & \textbf{0.129} \\
ROBIN            & 0.181 & \textbf{0.135} & 0.255 & \textbf{0.204} \\
PRC              & 0.182 & \textbf{0.117} & 0.297 & \textbf{0.193} \\
Tree-Ring        & 0.204 & \textbf{0.135} & 0.325 & \textbf{0.212} \\
RingID        & 0.440 & \textbf{0.291} & 0.586 & \textbf{0.453} \\
GaussMarker      & 0.575 & \textbf{0.372} & 0.690 & \textbf{0.479} \\
Gaussian Shading & 0.691 & \textbf{0.462} & never & \textbf{0.542} \\
\midrule
\textbf{Mean}    & 0.272& \textbf{0.184}& 0.318& \textbf{0.231}\\
\bottomrule
\end{tabular}
\end{table}

Stochastic samplers require $32\%$ less strength at $\lambda^\star_{50}$ and $27\%$ less at $\lambda^\star_5$ on average. The gap is negligible near the grid floor (SFW) but pronounced on robust schemes. Most decisively, deterministic Gaussian Shading never reaches $5\%$ survival by $\lambda=0.70$, whereas the stochastic condition does at about $\lambda=0.54$; this deterministic condition is the regeneration baseline in Remark~\ref{rem:stochasticity}. Fidelity is not uniformly higher: at each sampler's $\lambda^\star_5$, the mean SDE--ODE change is $-0.004$ SSIM and $-0.21$\,dB, although GaussMarker gains $+0.049$ SSIM and $+1.56$\,dB. The empirical advantage is therefore \emph{decoupling per unit strength}, not universally higher fidelity.

\paragraph{Latent-space evidence.}
\label{ssec:latent_evidence}
DDIM inversion of an unattacked image remains close to the watermark-carrying noise $\boldsymbol{\epsilon}^w$ ($L_1\approx0.34$--$0.68$), while the measured distance after DRIFT increases across the tested $\lambda$ grid toward the independent-Gaussian diagnostics $L_1=2/\sqrt{\pi}$ and $L_2=\sqrt2$. This empirical diagnostic complements, but is not implied by, the fixed-depth dependence bound in Theorem~\ref{prop:shift_informal}; Figure~\ref{fig:noise_distance} of Appendix~\ref{app:noise_distance} gives all nine curves. Corollary~\ref{cor:shift_2d} formalizes only the conditional squared-$L_2$ reference, not the empirical monotonicity or the Gaussian $L_1$ value.

\begin{figure*}[t]
\centering
\renewcommand{\arraystretch}{1.0}
\newlength{\methodcol}
\newlength{\imgcol}
\setlength{\methodcol}{0.055\textwidth}
\setlength{\imgcol}{0.18\textwidth}

\begin{minipage}[c]{\methodcol}\centering\end{minipage}%
\begin{minipage}[c]{\imgcol}\centering \scriptsize Original \end{minipage}%
\begin{minipage}[c]{\imgcol}\centering \scriptsize DRIFT \end{minipage}%
\begin{minipage}[c]{\imgcol}\centering \scriptsize Black-Box Attack \end{minipage}%
\begin{minipage}[c]{\imgcol}\centering \scriptsize Removing Attack \end{minipage}

\vspace{0.3mm}
\begin{minipage}[c]{\methodcol}\centering \small \textbf{GM} \end{minipage}%
\begin{minipage}[c]{\imgcol}\centering \zoomfig{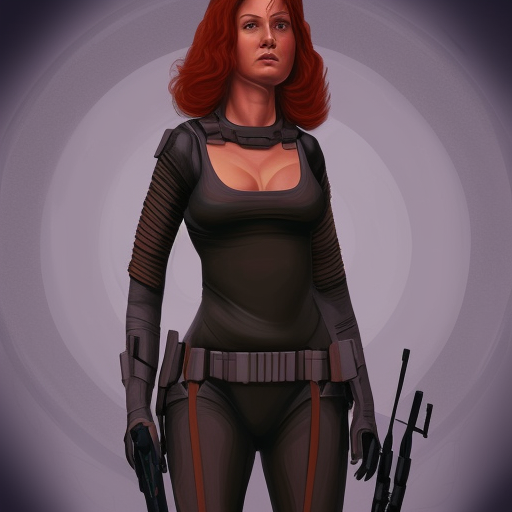}{5.0}{0.89}{0.63\linewidth}{0.85cm} \end{minipage}%
\begin{minipage}[c]{\imgcol}\centering \zoomfig{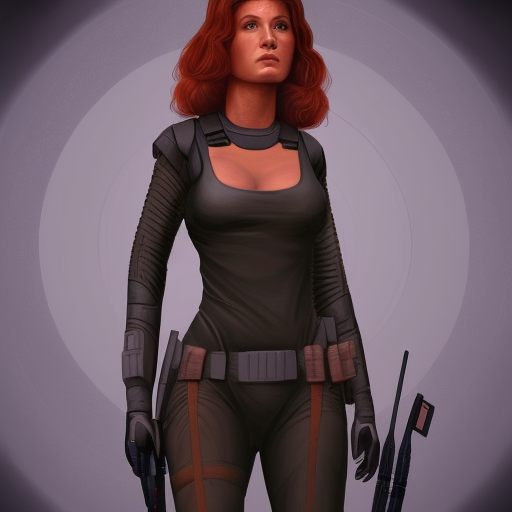}{5.0}{0.89}{0.63\linewidth}{0.85cm} \end{minipage}%
\begin{minipage}[c]{\imgcol}\centering \zoomfig{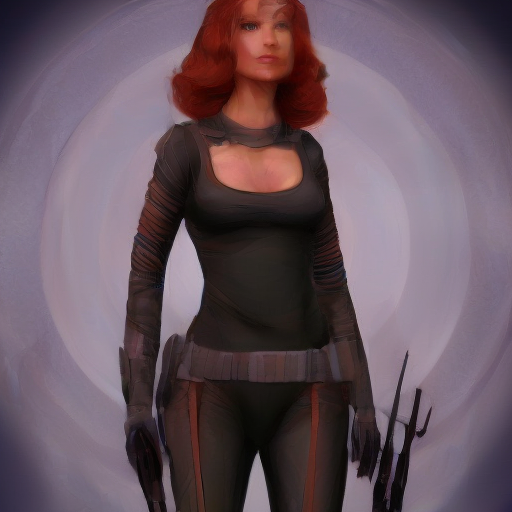}{5.0}{0.89}{0.63\linewidth}{0.85cm} \end{minipage}%
\begin{minipage}[c]{\imgcol}\centering \zoomfig{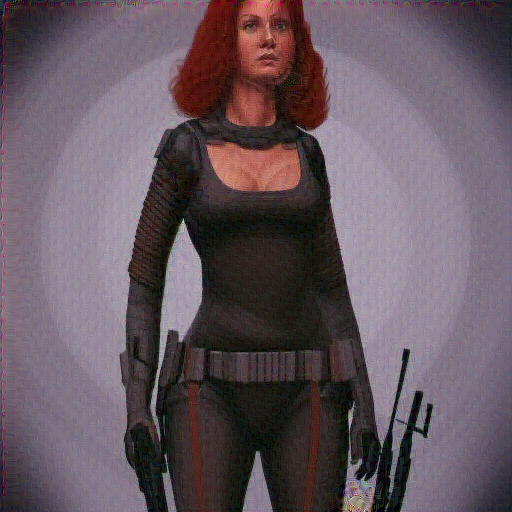}{5.0}{0.89}{0.63\linewidth}{0.85cm} \end{minipage}

\vspace{0.3mm}
\begin{minipage}[c]{\methodcol}\centering \small \textbf{GS} \end{minipage}%
\begin{minipage}[c]{\imgcol}\centering \zoomfig{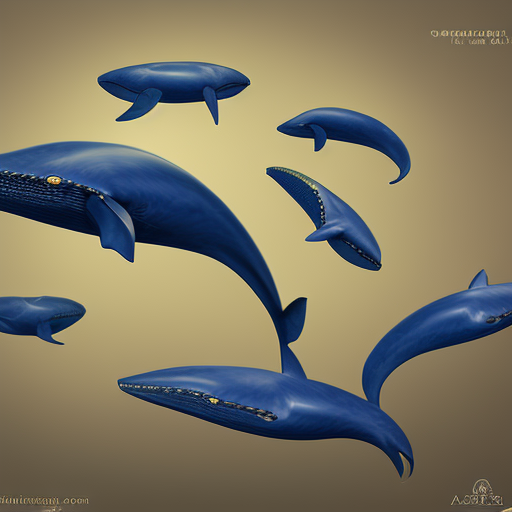}{0.72}{0.62}{0.63\linewidth}{0.85cm} \end{minipage}%
\begin{minipage}[c]{\imgcol}\centering \zoomfig{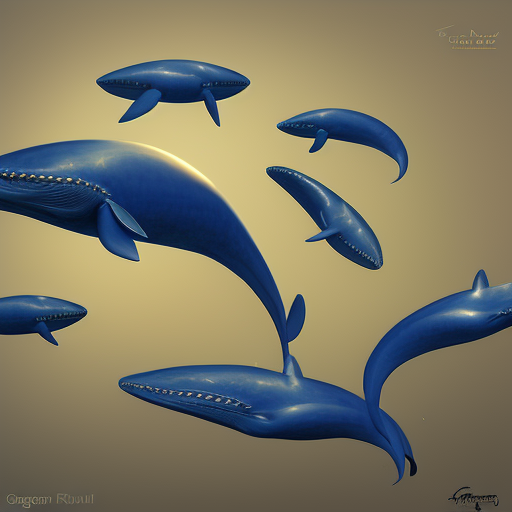}{0.72}{0.62}{0.63\linewidth}{0.85cm} \end{minipage}%
\begin{minipage}[c]{\imgcol}\centering \zoomfig{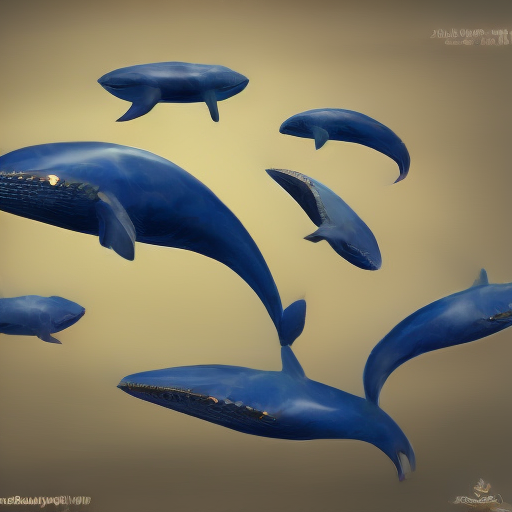}{0.72}{0.62}{0.63\linewidth}{0.85cm} \end{minipage}%
\begin{minipage}[c]{\imgcol}\centering \zoomfig{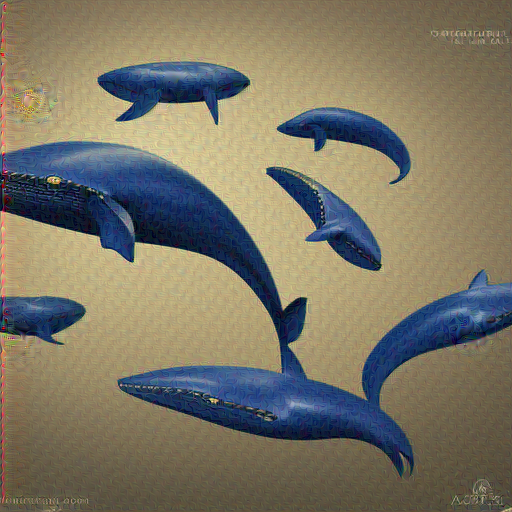}{0.72}{0.62}{0.63\linewidth}{0.85cm} \end{minipage}

\vspace{0.3mm}
\begin{minipage}[c]{\methodcol}\centering \small \textbf{PRC} \end{minipage}%
\begin{minipage}[c]{\imgcol}\centering \zoomfig{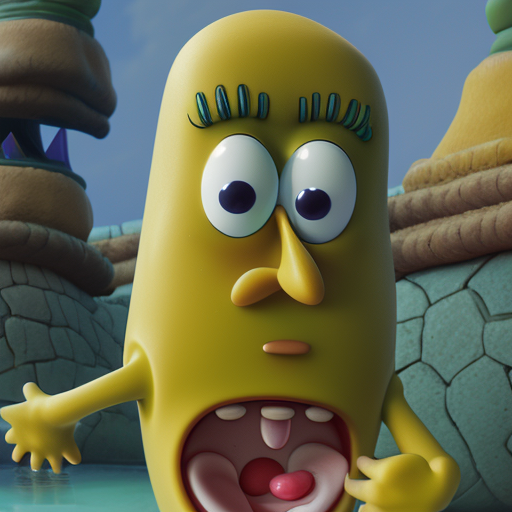}{2.0}{0.77}{0.63\linewidth}{0.85cm} \end{minipage}%
\begin{minipage}[c]{\imgcol}\centering \zoomfig{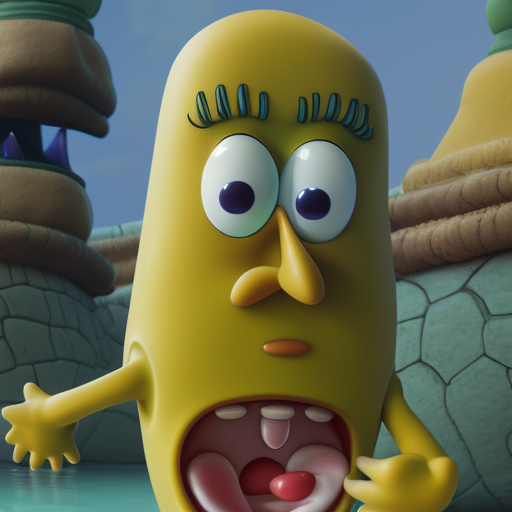}{2.0}{0.77}{0.63\linewidth}{0.85cm} \end{minipage}%
\begin{minipage}[c]{\imgcol}\centering \zoomfig{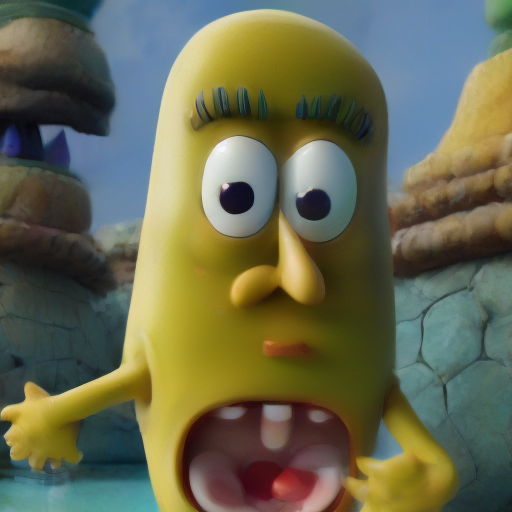}{2.0}{0.77}{0.63\linewidth}{0.85cm} \end{minipage}%
\begin{minipage}[c]{\imgcol}\centering \zoomfig{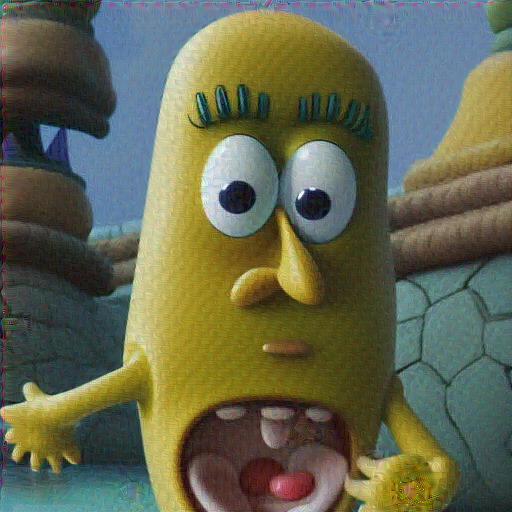}{2.0}{0.77}{0.63\linewidth}{0.85cm} \end{minipage}

\caption{\textbf{Qualitative comparison with baseline attacks} (Original / DRIFT / Black-Box / Removing), with zoomed insets, on GM. GS and PRC are in
Figure~\ref{fig:qualitative_zoom_appendix} of Appendix~\ref{app:qualitative_all}.}
\label{fig:qualitative_zoom_compare}
\end{figure*}

\subsection{Adaptive Search and Fidelity Refinement}

We next test blind family prediction, per-image strength selection, and verifier-gated refinement.

\paragraph{Blind family identification.}
\label{ssec:ident_results}
Using only structural features of the shared DDIM-inverted latent, the key-free classifier reaches \textbf{99.3\%} top-1 accuracy on $1{,}350$ held-out images. All non-trivial confusions stay within the similar Gaussian-Shading/GaussMarker/SEAL cluster. The prediction seeds verify-and-climb; per-family results and the confusion matrix appear in Table~\ref{tab:ident} and Figure~\ref{fig:confusion_blind} of Appendix~\ref{app:ident_details}.

\paragraph{Per-image minimal-strength search.}
\label{ssec:minimal_results}
On $300$ images per family, black-box verify-and-climb reduces mean strength from $0.244$ to $0.20$ relative to a fixed per-family choice while raising ASR from $96.8\%$ to $99.7\%$. Mean SSIM/PSNR/LPIPS improve from $0.67/22.6/0.19$ to $0.71/23.9/0.16$ at about $2.2$ verifier queries per image. This is consistent with Proposition~\ref{prop:minimal}'s ladder-relative result; Table~\ref{tab:adaptive_lambda} of Appendix~\ref{app:adaptive_details} gives full per-family results, while Figure~\ref{fig:robustness_compare} of Appendix~\ref{app:robustness_levels} visualizes the strength--fidelity trade-off.

\paragraph{Verifier-gated fidelity refinement.}
\label{ssec:refine_results}
On already rejected Stage-III inputs, DPPO improves every reported fidelity metric, including LPIPS from $0.136$ to $0.114$, while survival remains $0\%$ over $1.16$ rounds. Proposition~\ref{prop:feasible} explains the same-verifier rejection invariant, not the empirical fidelity gain. Table~\ref{tab:rl_ablation} shows that removing the terminal watermark bonus or back-off/re-verification raises raw survival to $9.3\%$ and $7.1\%$, while removing LPIPS guidance worsens LPIPS to $0.136$. Table~\ref{tab:refine} and Figures~\ref{fig:dppo_curves}--\ref{fig:refine_qual} of Appendix~\ref{app:rl_details} provide the full refinement evidence.

\begin{table}[t]
\centering
\small
\setlength{\tabcolsep}{3pt}
\renewcommand{\arraystretch}{1.08}
\begin{tabular}{@{}lcccc@{}}
\toprule
\textbf{Variant} & \textbf{SSIM$\uparrow$} & \textbf{LPIPS$\downarrow$} &
\textbf{WM surv.$\downarrow$} & \textbf{Rounds} \\
\midrule
Full DPPO (ours)                   & \textbf{0.738} & \textbf{0.114} & \textbf{0.0\%} & 1.16 \\
\midrule
$-$ terminal WM bonus$^{\ast}$    & 0.688 & 0.119 & 9.3\% & 1.00 \\
$-$ back-off / re-verify$^{\ast}$ & 0.694 & 0.117 & 7.1\% & 1.12 \\
$-$ LPIPS guidance                & 0.735 & 0.136 & 0.0\% & 1.22 \\
$k_{\max}=1$                      & 0.738 & 0.116 & 0.0\% & 1.00 \\
\bottomrule
\end{tabular}
\caption{\textbf{Ablation of the DPPO controller} ($N=225$). Each row disables
one component. $^{\ast}$For the two constraint ablations, survival is measured
with evaluation-time back-off disabled; otherwise deployed survival remains
fixed at $0\%$.}
\label{tab:rl_ablation}
\end{table}

\subsection{Image Quality and Visual Fidelity}
\label{ssec:e2e}

Table~\ref{tab:main_results}(b) reports semantic, distributional, and paired fidelity over all nine families. Fixed $\lambda=0.70$ reaches $100\%$ ASR but incurs substantial distortion (SSIM $0.401$, LPIPS $0.528$); per-image selection improves to SSIM $0.697$/LPIPS $0.168$ at $99.7\%$ ASR. Full DRIFT restores $100\%$ ASR with SSIM $0.713$, PSNR $24.01$, LPIPS $0.154$, FID $55.414$, and $8.12$\,s runtime. Its FID/LPIPS are lower than both baselines, while the removing attack reaches only $81.3\%$ ASR; Table~\ref{tab:attack_compare} of Appendix~\ref{app:quality} gives the aggregate comparison and protocol. Figure~\ref{fig:qualitative_zoom_compare} provides the retained zoomed comparison on GM, GS, and PRC; Figures~\ref{fig:qualitative_all_methods_a}--\ref{fig:qualitative_all_methods_b} of Appendix~\ref{app:qualitative_all} cover all nine families.

%% file: conclusion.tex
\section{Conclusion}\label{sec:conclusion}

We identify trajectory consistency as a common vulnerability among the diffusion
watermarks studied and introduce DRIFT, which combines partial forward
diffusion, stochastic reverse resampling, per-image verify-and-climb, and
verifier-gated refinement. Across nine watermarks in three paradigms, DRIFT
achieves $98$--$100\%$ ASR and the best perceptual-quality trade-off among the
compared attacks without secret keys, verifier internals, or per-image
gradients. Matched samplers isolate the empirical benefit of stochastic
deflection; Theorem~\ref{prop:shift_informal} and
Proposition~\ref{prop:feasible} separately characterize fixed-depth source
dependence and same-verifier refinement feasibility. Aggressive fixed
re-noising can substantially reduce fidelity, which motivates the adaptive
minimal-strength search and refinement rather than a universally strong
setting. The guarantees remain conditional on unverified network-level
Lipschitz and terminal premises, and evaluation uses one public latent-diffusion
backbone. Extending trajectory-aware evaluation to pixel-space, video,
autoregressive, and flow-based generators is an important next step.

%% file: appendix.tex
\appendix

\begin{center}
{\LARGE\bfseries Technical Appendix\par}
\end{center}
\vspace{0.5em}

\noindent\textbf{Organization.}
Appendix~\ref{app:algorithms} gives executable specifications of the base and
adaptive attacks; Appendix~\ref{app:theory} contains the full theoretical
statements and proofs; Appendices~\ref{app:adaptive} and~\ref{app:rl} prove the
adaptive-search and refinement guarantees; and
Appendix~\ref{app:experiments} provides the extended protocols, tables,
ablations, and qualitative results.

\section{Algorithmic Specifications}
\label{app:algorithms}

Algorithms~\ref{alg:drift} and~\ref{alg:adaptive} provide executable summaries of the base and adaptive attacks described in Section~\ref{sec:method}.

\begin{algorithm}[H]
\caption{DRIFT (base attack)}
\label{alg:drift}
\begin{algorithmic}[1]
\Require Watermarked image $\mathbf{x}^w$, pretrained LDM $(\mathcal{E},\mathcal{D},\boldsymbol{\epsilon}_\theta)$, global source-independent conditioning $c_{\mathrm{att}}$, strength $\lambda$, schedule $\{\beta_i\}_{i=1}^T$, valid generalized reverse scales $\{\sigma_t\}$
\Ensure Attacked image $\mathbf{x}^a$
\State $\mathbf{z}_0^w\gets\mathcal{E}(\mathbf{x}^w)$; \quad $t_\lambda\gets\lceil\lambda T\rceil$
\State $\bar{\alpha}_0\gets1$; \quad $\bar{\alpha}_t\gets\prod_{i=1}^{t}(1-\beta_i)$
\State $\boldsymbol{\epsilon}\sim\mathcal{N}(\mathbf{0},\mathbf{I}_d)$
\State $\mathbf{z}_{t_\lambda}^a\gets\sqrt{\bar{\alpha}_{t_\lambda}}\mathbf{z}_0^w+\sqrt{1-\bar{\alpha}_{t_\lambda}}\boldsymbol{\epsilon}$
\For{$t=t_\lambda,t_\lambda-1,\ldots,1$}
    \State $\hat{\mathbf{z}}_0^a\gets\bigl(\mathbf{z}_t^a-\sqrt{1-\bar{\alpha}_t}\boldsymbol{\epsilon}_\theta(\mathbf{z}_t^a,t;c_{\mathrm{att}})\bigr)/\sqrt{\bar{\alpha}_t}$
    \State $\boldsymbol{\xi}_t\sim\mathcal{N}(\mathbf{0},\mathbf{I}_d)$
    \State $\mathbf{z}_{t-1}^a\gets\sqrt{\bar{\alpha}_{t-1}}\hat{\mathbf{z}}_0^a+\sqrt{1-\bar{\alpha}_{t-1}-\sigma_t^2}\boldsymbol{\epsilon}_\theta(\mathbf{z}_t^a,t;c_{\mathrm{att}})+\sigma_t\boldsymbol{\xi}_t$
\EndFor
\State \Return $\mathbf{x}^a\gets\mathcal{D}(\mathbf{z}_0^a)$
\end{algorithmic}
\end{algorithm}

\begin{algorithm}[H]
\caption{Adaptive DRIFT (predict $\to$ seed $\to$ climb)}
\label{alg:adaptive}
\begin{algorithmic}[1]
\Require $\mathbf{x}^w$; encoder/inverter/identifier/verifier $(\mathcal{E},\mathcal{I}_{\mathrm{DDIM}},\mathcal{H},\mathcal{V})$; global $c_{\mathrm{att}}$; ladder bank $\{\Lambda_g\}_{g\in\mathcal{F}}$
\Ensure Attacked image $\mathbf{x}^a$, selected strength $\lambda$, and status
\State $\mathbf{z}_T^{\mathrm{inv}}\gets\mathcal{I}_{\mathrm{DDIM}}(\mathcal{E}(\mathbf{x}^w);c_{\mathrm{att}})$; \quad $\hat f\gets\mathcal{H}(\mathbf{z}_T^{\mathrm{inv}})$
\State Select the ascending seeded ladder $\Lambda_{\hat f}$
\For{$\lambda\in\Lambda_{\hat f}$}
    \State $\mathbf{x}^a\gets\textsc{DRIFT}(\mathbf{x}^w,\lambda)$ \Comment{Algorithm~\ref{alg:drift}}
    \If{$\mathcal{V}(\mathbf{x}^a)$ reports \textsc{not watermarked}}
        \State \Return $(\mathbf{x}^a,\lambda,\textsc{Success})$
        \Comment{first verifier-rejected rung}
    \EndIf
\EndFor
\State \Return $(\mathbf{x}^a,\lambda,\textsc{Failure})$
\Comment{strongest evaluated candidate remains detected}
\end{algorithmic}
\end{algorithm}

\section{Theory of Source Decoupling}
\label{app:theory}

We prove the three parts of Theorem~\ref{prop:shift_informal}.  Throughout,
$\mathbf X=\mathbf z_0^w\in\mathbb R^d$ is the source latent,
$\mathbf W=\boldsymbol{\epsilon}^w
=\mathcal I_{\mathrm{DDIM}}(\mathbf X)\in\mathbb R^d$ is its recovered
reference, and $n=t_\lambda=\lceil\lambda T\rceil$ is the attack depth.
For the distributional claims, $(\mathbf X,\mathbf W)$ follows the population
law induced by the watermarked-image generation protocol, and is independent
of the attack randomness.  When $\mathbf X=\mathbf x$ is fixed, expectations
subscripted by $\mathcal N_n$ are over attack randomness only.

\begin{assumption}[Regularity]
\label{asm:shift_regular}
We assume:
(A1) $0<\beta_t<1$ for every diffusion step,
$\mathbf X,\mathbf W\in L^2$, the recovered reference
$\mathbf W=\mathcal I_{\mathrm{DDIM}}(\mathbf X)$ is a measurable function
of $\mathbf X$, and $0<\bar\alpha_n<1$;
(A2) the entries of the complete attack-noise vector
$\mathcal N_n=(\boldsymbol{\epsilon},\boldsymbol{\xi}_{1:n})$ are mutually
independent standard Gaussian draws, and $\mathcal N_n$ is jointly independent
of $(\mathbf X,\mathbf W)$;
(A3) all denoising and inversion operations use the same globally fixed
deterministic attacker-side conditioning $c_{\mathrm{att}}$, independent of
$(\mathbf X,\mathbf W,\mathcal N_n)$, and all forward, reverse, decoder,
encoder, and inversion maps used below are Borel measurable;
(A4) for the sensitivity claims, each fixed-noise reverse kernel is
$K_t$-Lipschitz in its state uniformly over the reverse-noise argument, the
recovered-noise map
$Q:=\mathcal I_{\mathrm{DDIM}}\circ\mathcal E\circ\mathcal D$ is
$L_Q$-Lipschitz, and the source-free recovered output defined below is in
$L^2$.
\end{assumption}

Set $\bar\alpha_0=1$ and
$0\le\sigma_t\le\sqrt{1-\bar\alpha_{t-1}}$.  Substituting
Eq.~\eqref{eq:predict_clean} into Eq.~\eqref{eq:reverse} gives
\begin{equation}
\mathbf z_{t-1}
=g_t(\mathbf z_t;\boldsymbol{\xi}_t)
=A_t\mathbf z_t
+B_t\boldsymbol{\epsilon}_\theta(\mathbf z_t,t)
+\sigma_t\boldsymbol{\xi}_t,
\label{eq:appendix_affine}
\end{equation}
where
\[
A_t=\sqrt{\frac{\bar\alpha_{t-1}}{\bar\alpha_t}},
\qquad
B_t=
\sqrt{1-\bar\alpha_{t-1}-\sigma_t^2}
-\sqrt{\frac{\bar\alpha_{t-1}(1-\bar\alpha_t)}
{\bar\alpha_t}}.
\]
This is exactly the generalized DDIM update; $\sigma_t=0$ gives
deterministic DDIM, and the DDPM posterior variance gives the standard
ancestral DDPM special case.  For fixed $\boldsymbol{\xi}_{1:n}$, let
$F_n(\mathbf u;\boldsymbol{\xi}_{1:n})$ be the reverse composition from
time $n$ to $0$, and set
$R_n=Q\circ F_n$.  With
$\boldsymbol{\Xi}_{1:n}:=(\boldsymbol{\xi}_1,\ldots,\boldsymbol{\xi}_n)$,
the main-text map is
$H_n(\mathbf u,\boldsymbol{\Xi}_{1:n})
:=R_n(\mathbf u;\boldsymbol{\xi}_{1:n})$.  Define
\[
C_n=\prod_{s=1}^nK_s,\qquad
\Gamma_n=\sqrt{\bar\alpha_n}C_n,\qquad
\Delta_n=L_Q^2\Gamma_n^2\,\mathbb E\|\mathbf X\|_2^2.
\]
If $\boldsymbol{\epsilon}_\theta(\cdot,t)$ is $L_t$-Lipschitz, then the
explicit but generally loose choice
$K_t=A_t+|B_t|L_t$ is valid.  Since this triangle bound ignores cancellation,
it need not certify that $\Gamma_n$ decreases; the information result below
does not require any Lipschitz assumption.

\begin{lemma}[Forward information bottleneck]
\label{lem:shift_information}
Under (A1)--(A3), with
$\mathbf Z_n=\sqrt{\bar\alpha_n}\mathbf X+
\sqrt{1-\bar\alpha_n}\boldsymbol{\epsilon}$ and
$\mathbf Y_n=H_n(\mathbf Z_n,\boldsymbol{\Xi}_{1:n})$, where
$H_n$ is the complete measurable reverse and recovered-noise pipeline,
\[
I(\mathbf W;\mathbf Y_n)
\le
\frac12\log\det\!\left(
\mathbf I_d+
\frac{\bar\alpha_n}{1-\bar\alpha_n}\Sigma_z\right)
\le\mathcal B_n,
\]
where $\Sigma_z=\operatorname{Cov}(\mathbf X)$ and $\mathcal B_n$ is defined
in Eq.~\eqref{eq:information_bound}.  Moreover,
\[
\operatorname{TV}\!\left(
\mathcal L(\mathbf W,\mathbf Y_n),
\mathcal L(\mathbf W)\otimes\mathcal L(\mathbf Y_n)\right)
\le\sqrt{\mathcal B_n/2}.
\]
The envelope $\mathcal B_n$ is non-increasing with attack depth.
\end{lemma}
\begin{proof}
Freshness gives the Markov chains
$\mathbf W\to\mathbf X\to\mathbf Z_n\to\mathbf Y_n$, hence
$I(\mathbf W;\mathbf Y_n)\le I(\mathbf X;\mathbf Z_n)$.
For $a=\bar\alpha_n$ and $\Sigma_z=\operatorname{Cov}(\mathbf X)$, the
maximum-entropy property of the Gaussian gives
\[
I(\mathbf X;\mathbf Z_n)
\le\frac12\log\det\!\left(
\mathbf I_d+\frac{a}{1-a}\Sigma_z\right).
\]
Applying Jensen's inequality to the eigenvalues of $\Sigma_z$ gives
Eq.~\eqref{eq:information_bound}, and Pinsker's inequality gives the TV
bound.  Finally, if $m\ge n$, then
$\bar\alpha_m\le\bar\alpha_n$; equivalently, the coupling
\[
\mathbf Z_m=
\sqrt{\frac{\bar\alpha_m}{\bar\alpha_n}}\,\mathbf Z_n
+\sqrt{1-\frac{\bar\alpha_m}{\bar\alpha_n}}\,\mathbf G,
    \qquad \mathbf G\sim\mathcal N(\mathbf0,\mathbf I_d),
    \quad \mathbf G\perp(\mathbf W,\mathbf Z_n),
\]
makes $\mathbf W\to\mathbf Z_n\to\mathbf Z_m$ a Markov chain.  Both views
show that the forward information envelope is non-increasing with depth.
\end{proof}

\begin{lemma}[One-step stability]
\label{lem:shift_one_step}
For every $t$, $\mathbf{u},\mathbf{v}$, and $\boldsymbol{\xi}$,
\[
\|g_t(\mathbf u;\boldsymbol{\xi})
-g_t(\mathbf v;\boldsymbol{\xi})\|_2
\le K_t\|\mathbf u-\mathbf v\|_2 .
\]
\end{lemma}
\begin{proof}
This is (A4).  For the explicit sufficient constant, the common noise
cancels and the triangle inequality gives
\[
\|g_t(\mathbf u;\boldsymbol{\xi})
-g_t(\mathbf v;\boldsymbol{\xi})\|_2
\le(A_t+|B_t|L_t)\|\mathbf u-\mathbf v\|_2 .
\]
\end{proof}

\begin{lemma}[Multi-step stability]
\label{lem:shift_multistep}
For every $n$, $\mathbf{u},\mathbf{v}$, and $\boldsymbol{\xi}_{1:n}$, $\|F_n(\mathbf{u};\boldsymbol{\xi}_{1:n}) - F_n(\mathbf{v};\boldsymbol{\xi}_{1:n})\|_2 \le C_n \|\mathbf{u}-\mathbf{v}\|_2$.
\end{lemma}
\begin{proof}
Drive both trajectories by the same noises and iterate
Lemma~\ref{lem:shift_one_step} from $t=n$ to $1$:
\[
\|\mathbf z_0^{(\mathbf u)}-\mathbf z_0^{(\mathbf v)}\|_2
\le\left(\prod_{s=1}^nK_s\right)\|\mathbf u-\mathbf v\|_2
=C_n\|\mathbf u-\mathbf v\|_2.
\]
\end{proof}

\begin{lemma}[Stability of the recovered-noise map]
\label{lem:shift_recovered_stability}
$\|R_n(\mathbf{u};\boldsymbol{\xi}_{1:n}) - R_n(\mathbf{v};\boldsymbol{\xi}_{1:n})\|_2 \le L_Q C_n \|\mathbf{u}-\mathbf{v}\|_2$.
\end{lemma}
\begin{proof}
By (A4),
$\|R_n(\mathbf u)-R_n(\mathbf v)\|_2
\le L_Q\|F_n(\mathbf u)-F_n(\mathbf v)\|_2$; apply
Lemma~\ref{lem:shift_multistep}.
\end{proof}

\begin{lemma}[Decoupling from the source latent]
\label{lem:shift_terminal_decoupling}
Let
\[
\widehat{\boldsymbol{\epsilon}}_n^a
=R_n\!\left(
\sqrt{\bar\alpha_n}\mathbf X+
\sqrt{1-\bar\alpha_n}\boldsymbol{\epsilon};
\boldsymbol{\xi}_{1:n}\right),
\qquad
\widetilde{\boldsymbol{\epsilon}}_n
=R_n\!\left(
\sqrt{1-\bar\alpha_n}\boldsymbol{\epsilon};
\boldsymbol{\xi}_{1:n}\right).
\]
Then $\widetilde{\boldsymbol{\epsilon}}_n$ is independent of
$(\mathbf X,\mathbf W)$.  Moreover, for every deterministic
$\mathbf x\in\mathbb R^d$ and every attack-noise realization,
\[
\left\|
R_n\!\left(
\sqrt{\bar\alpha_n}\mathbf x+
\sqrt{1-\bar\alpha_n}\boldsymbol{\epsilon};
\boldsymbol{\xi}_{1:n}\right)
-\widetilde{\boldsymbol{\epsilon}}_n
\right\|_2
\le L_QC_n\sqrt{\bar\alpha_n}\|\mathbf x\|_2 .
\]
Consequently,
\[
\mathbb E_{\mathcal N_n}\!\left[
\left\|
R_n\!\left(
\sqrt{\bar\alpha_n}\mathbf x+
\sqrt{1-\bar\alpha_n}\boldsymbol{\epsilon};
\boldsymbol{\xi}_{1:n}\right)
-\widetilde{\boldsymbol{\epsilon}}_n
\right\|_2^2
\right]
\le L_Q^2C_n^2\bar\alpha_n\|\mathbf x\|_2^2,
\]
and averaging over the source population gives
\[
\mathbb E\|
\widehat{\boldsymbol{\epsilon}}_n^a-
\widetilde{\boldsymbol{\epsilon}}_n\|_2^2
\le \Delta_n .
\]
\end{lemma}
\begin{proof}
The baseline is a measurable function of $\mathcal N_n$ and the globally fixed
$c_{\mathrm{att}}$ alone, so independence follows from (A2)--(A3).  Couple
both runs with the same $\mathcal N_n$.
Lemma~\ref{lem:shift_recovered_stability} gives the displayed deterministic
$\mathbf x$ bound pathwise.  Squaring and integrating over $\mathcal N_n$
gives the fixed-image bound.  Substituting $\mathbf x=\mathbf X$ gives,
pathwise,
\[
\|
\widehat{\boldsymbol{\epsilon}}_n^a-
\widetilde{\boldsymbol{\epsilon}}_n\|_2
\le L_QC_n\sqrt{\bar\alpha_n}\|\mathbf X\|_2 .
\]
Squaring and averaging over the independent source population proves the last
claim.
\end{proof}

\begin{lemma}[Distance to the independent product law]
\label{lem:shift_product_law}
All laws below belong to $\mathcal P_2(\mathbb R^{2d})$, and
\[
W_2\!\left(
\mathcal L(\widehat{\boldsymbol{\epsilon}}_n^a,\mathbf W),
\mathcal L(\widehat{\boldsymbol{\epsilon}}_n^a)
\otimes\mathcal L(\mathbf W)\right)
\le 2\sqrt{\Delta_n}.
\]
\end{lemma}
\begin{proof}
Write $A=\widehat{\boldsymbol{\epsilon}}_n^a$,
$B=\widetilde{\boldsymbol{\epsilon}}_n$, and $E=\mathbf W$.
Assumptions (A1), (A4), and Lemma~\ref{lem:shift_terminal_decoupling} give the
required second moments.  Since $B\perp E$,
\[
\rho:=\mathcal L(B,E)=\mathcal L(B)\otimes\mathcal L(E).
\]
The synchronous coupling $((A,E),(B,E))$ shows
$W_2(\mathcal L(A,E),\rho)\le\sqrt{\Delta_n}$.
For the other leg, draw $(A',B')\sim\mathcal L(A,B)$ and independently draw
$E'\sim\mathcal L(E)$.  Then $((B',E'),(A',E'))$ couples
$\rho$ to $\mathcal L(A)\otimes\mathcal L(E)$ at expected squared cost
$\mathbb E\|A-B\|_2^2\le\Delta_n$.  The triangle inequality yields the stated
factor of two.
\end{proof}

\begin{theorem}[Source dependence at fixed depth]
\label{prop:noise_distance_formal}
Fix a deterministic $\lambda\in(0,1]$ before sampling
$(\mathbf X,\mathbf W,\mathcal N_n)$ and set $n=t_\lambda$.
Under Assumption~\ref{asm:shift_regular}: (i) the information and
total-variation bounds of
Lemma~\ref{lem:shift_information} hold; and (ii) the pathwise and fixed-image
bounds of Lemma~\ref{lem:shift_terminal_decoupling} hold, while population
averaging gives
\[
\mathbb E\|
\widehat{\boldsymbol{\epsilon}}_n^a-
\widetilde{\boldsymbol{\epsilon}}_n\|_2^2
\le\Delta_n,
\qquad
W_2\!\left(
\mathcal L(\widehat{\boldsymbol{\epsilon}}_n^a,\mathbf W),
\mathcal L(\widehat{\boldsymbol{\epsilon}}_n^a)\otimes\mathcal L(\mathbf W)
\right)
\le2\sqrt{\Delta_n}.
\]
Only the information envelope $\mathcal B_n$ is guaranteed to be
non-increasing across prespecified depths; no monotonicity is asserted for
$\Delta_n$.  These distributional bounds do not automatically extend to the
output selected by a source- or verifier-dependent adaptive depth.
\end{theorem}
\begin{proof}
Combine Lemmas~\ref{lem:shift_information},
\ref{lem:shift_terminal_decoupling}, and~\ref{lem:shift_product_law}.
\end{proof}

\begin{assumption}[Terminal reference moments]
\label{asm:shift_prior}
The source-free terminal baseline
$\widetilde{\boldsymbol{\epsilon}}_T$ and reference $\mathbf W$ are centered
and have identity covariance.
\end{assumption}

Together with (A2), this moment condition implies
$\mathbb E\|
\widetilde{\boldsymbol{\epsilon}}_T-\mathbf W\|_2^2=2d$; Gaussianity is not
needed.  The condition is an explicit diagnostic premise, not a consequence
of Gaussian primitive attack noise or of a standard marginal for a
watermark-carrying latent.  It must therefore be justified or measured for
the recovered-noise pipeline in question.

\begin{corollary}[Random-reference distance at the terminal step]
\label{cor:shift_2d}
When $n=T$, under Assumption~\ref{asm:shift_regular} instantiated at $T$
and Assumption~\ref{asm:shift_prior},
\[
\left|
\mathbb E\|
\widehat{\boldsymbol{\epsilon}}_T^a-\mathbf W\|_2^2-2d
\right|
\le \Delta_T+2\sqrt{2d\,\Delta_T}.
\]
\end{corollary}
\begin{proof}
Set $U=\widehat{\boldsymbol{\epsilon}}_T^a-
\widetilde{\boldsymbol{\epsilon}}_T$ and
$V=\widetilde{\boldsymbol{\epsilon}}_T-\mathbf W$.  Expanding
$\|U+V\|_2^2$ and applying Cauchy--Schwarz gives
\[
\left|
\mathbb E\|U+V\|_2^2-\mathbb E\|V\|_2^2
\right|
\le \mathbb E\|U\|_2^2+
2\sqrt{\mathbb E\|U\|_2^2\,\mathbb E\|V\|_2^2}.
\]
Theorem~\ref{prop:noise_distance_formal} bounds the first moment by
$\Delta_T$.  By Lemma~\ref{lem:shift_terminal_decoupling}, (A2), and
Assumption~\ref{asm:shift_prior}, the centered vectors
$\widetilde{\boldsymbol{\epsilon}}_T$ and $\mathbf W$ are independent and
$\mathbb E\|V\|_2^2=2d$.  Substitution proves the result.
\end{proof}

\begin{lemma}[When stochastic paths remain diverse]
\label{lem:shift_path_diversity}
Fix an initial state $\mathbf z_n$ and let
$G(\mathbf z_n,\boldsymbol{\Xi}_{1:n})$ denote the full reverse map into any
finite-dimensional latent space.  If
$\boldsymbol{\Xi}_{1:n}'$ is an independent copy, then the terminal law is
non-Dirac if and only if
\[
\Pr\!\left[
G(\mathbf z_n,\boldsymbol{\Xi}_{1:n})
\ne
G(\mathbf z_n,\boldsymbol{\Xi}_{1:n}')\right]>0.
\]
The same criterion holds after a Borel decoder by replacing $G$ with
$\mathcal D\circ G$.  A step with $\sigma_t>0$ has a non-Dirac conditional
next-state law, but this alone does not imply either terminal or decoded
diversity.
\end{lemma}
\begin{proof}
Two independent draws from a common probability law agree almost surely if
and only if that law is Dirac; applying this fact to the Borel random variable
$G(\mathbf z_n,\boldsymbol{\Xi}_{1:n})$ proves the equivalence.  At a positive
scale, the next state is a fixed conditional mean plus the non-degenerate
Gaussian $\sigma_t\boldsymbol{\xi}_t$, so its conditional law is non-Dirac.
A later measurable map, however, may map all of that variation to one point.
\end{proof}

\section{Analysis of Adaptive Search}
\label{app:adaptive}

\begin{proof}[Proof of Proposition~\ref{prop:minimal}]
Algorithm~\ref{alg:adaptive} scans the fixed ordered ladder from $j=0$ and
stops at its first rejected rung.  By definition this is $j^\star$, so the
scan uses exactly $j^\star+1$ verifier queries.  Every other rejected rung has
index $j\ge j^\star$; the distortion ordering in
Assumption~\ref{asm:monotone} gives
$\phi_j\ge\phi_{j^\star}$.  The comparison is only among candidates in this
same realized ladder.
\end{proof}

\section{Analysis of RL Fidelity Refinement}
\label{app:rl}

Fix the deterministic verifier and the score $S$ from
Eq.~\eqref{eq:fidelity_score} used in the rollout.
Let $v_k^{\mathrm{ref}}\in\{0,1\}$ be its decision at round $k$, where zero
means rejected.  Starting from the rejected $b_0$, update
\[
b_k=
\begin{cases}
\mathbf x_k,
&v_k^{\mathrm{ref}}=0
\ \text{and}\ S(\mathbf x_k)>S(b_{k-1}),\\
b_{k-1},&\text{otherwise}.
\end{cases}
\]

\begin{proof}[Proof of Proposition~\ref{prop:feasible}]
Induct on $k$.  The fixed verifier rejects $b_0$.  If it rejects $b_{k-1}$,
then either the update retains $b_{k-1}$ or replaces it by a candidate that
the same verifier rejects.  Hence every retained best, and thus the output of
every nested prefix of this rollout, remains rejected independently of the
policy.
\end{proof}

\begin{proof}[Proof of Proposition~\ref{prop:monotone_fid}]
At each round, either $b_k=b_{k-1}$ or the update occurs under the strict
condition $S(b_k)>S(b_{k-1})$.  Thus
$S(b_k)\ge S(b_{k-1})$ pathwise for every nested prefix of the fixed rollout.
\end{proof}

\noindent\textbf{On DPPO optimization.}
The diffusion policy's action log-probability is the sum of Gaussian
log-probabilities over the $K_{\mathrm{diff}}$ denoising sub-steps, so the
importance ratio
$r_\theta=\exp(\log\pi_\theta-\log\pi_{\theta_{\mathrm{old}}})$ is well
defined.  DPPO optimizes the PPO clipped surrogate
$\mathbb E[\min(r_\theta\widehat A,
\operatorname{clip}(r_\theta,1{-}\varepsilon_{\mathrm{clip}},
1{+}\varepsilon_{\mathrm{clip}})\widehat A)]$ with GAE advantage estimates
$\widehat A$~\citep{schulman2017ppo,schulman2016gae}.  This objective
regularizes large policy updates but is not a monotone-improvement guarantee
for true return.  Propositions~\ref{prop:feasible}--\ref{prop:monotone_fid}
instead follow from verifier-gated retention and hold independently of DPPO
convergence.

\section{Additional Experimental Results}
\label{app:experiments}

This appendix provides supplementary tables, ablations, and qualitative results that complement the main experiments.

% -------------------------------------------------------
\subsection{Image Quality Across Attack Methods}
\label{app:quality}
Table~\ref{tab:attack_compare} reports CLIP, FID, SSIM, and LPIPS averaged over the nine watermark families for the two baseline attacks and the full Adaptive+RL DRIFT pipeline. DRIFT attains the best score on every metric: the highest CLIP ($32.091$) and lowest FID ($55.41$), the best SSIM ($0.713$, matching the removing attack), and, most tellingly, the lowest LPIPS ($0.154$)---roughly half that of either baseline. The removing attack attains similar SSIM but substantially worse LPIPS and FID, showing that pixel similarity alone does not capture the full fidelity trade-off; DRIFT achieves higher attack success with less measured perceptual distortion.

\begin{table}[t]
\centering
\caption{\textbf{Image quality comparison across attack methods.} CLIP score, FID, SSIM, and LPIPS are averaged over the nine watermark families ($100$ images each). FID is the per-family-mean Fr\'echet distance to the watermarked originals; SSIM and LPIPS are paired against the same originals. Best per column in \textbf{bold}.}
\label{tab:attack_compare}
\setlength{\tabcolsep}{5pt}
\renewcommand{\arraystretch}{1.15}
\small
\begin{tabular}{lcccc}
\toprule
\textbf{Method Group} & \textbf{CLIP $\uparrow$} & \textbf{FID $\downarrow$} & \textbf{SSIM $\uparrow$} & \textbf{LPIPS $\downarrow$} \\
\midrule
Black-box attack & 31.877 & 114.019 & 0.519 & 0.305 \\
Removing attack  & 31.357 & 104.693 & 0.712 & 0.338 \\
DRIFT (Adaptive+RL) & \textbf{32.091} & \textbf{55.414} & \textbf{0.713} & \textbf{0.154} \\
\bottomrule
\end{tabular}
\end{table}

% -------------------------------------------------------
\subsection{Blind Watermark-Family Identification}
\label{app:ident_details}
Table~\ref{tab:ident} gives the per-family top-1 accuracy and low-confidence rate of the key-free classifier over $1{,}350$ held-out images ($150$ per family), and Figure~\ref{fig:confusion_blind} shows the row-normalised confusion matrix. Four families (Tree-Ring, WIND, SFW, ROBIN) are identified perfectly, and most confusions involve the structurally similar GS/GM/SEAL families, with a few isolated errors involving RI and PRC. 

\begin{table}[t]
\centering
\small
\caption{Blind watermark-family identification on held-out images (150 per
family). The key-free classifier predicts the family from the DDIM-inverted
latent $\mathbf{z}_T^{\mathrm{inv}}$. Top-1 is per-family recall; \emph{low-conf.}\ flags
confidence below $\tau_{\mathrm{conf}}{=}0.835$ (5th validation percentile).}
\label{tab:ident}
\begin{tabular}{lcc}
\toprule
Watermark family & Top-1 Acc.\ (\%) & Low-conf.\ (\%) \\
\midrule
Tree-Ring        & 100.0 & 2.7 \\
RingID           & 99.3  & 4.0 \\
PRC              & 99.3  & 2.7 \\
WIND             & 100.0 & 0.0 \\
Gaussian Shading & 97.3  & 12.7 \\
GaussMarker      & 99.3  & 7.3 \\
SFW              & 100.0 & 3.3 \\
SEAL             & 98.0  & 17.3 \\
ROBIN            & 100.0 & 0.0 \\
\midrule
\textbf{Overall (blind)}            & \textbf{99.3} & \textbf{5.6} \\
\bottomrule
\end{tabular}
\end{table}

\begin{figure}[!htbp]
\centering
\includegraphics[width=0.68\linewidth]{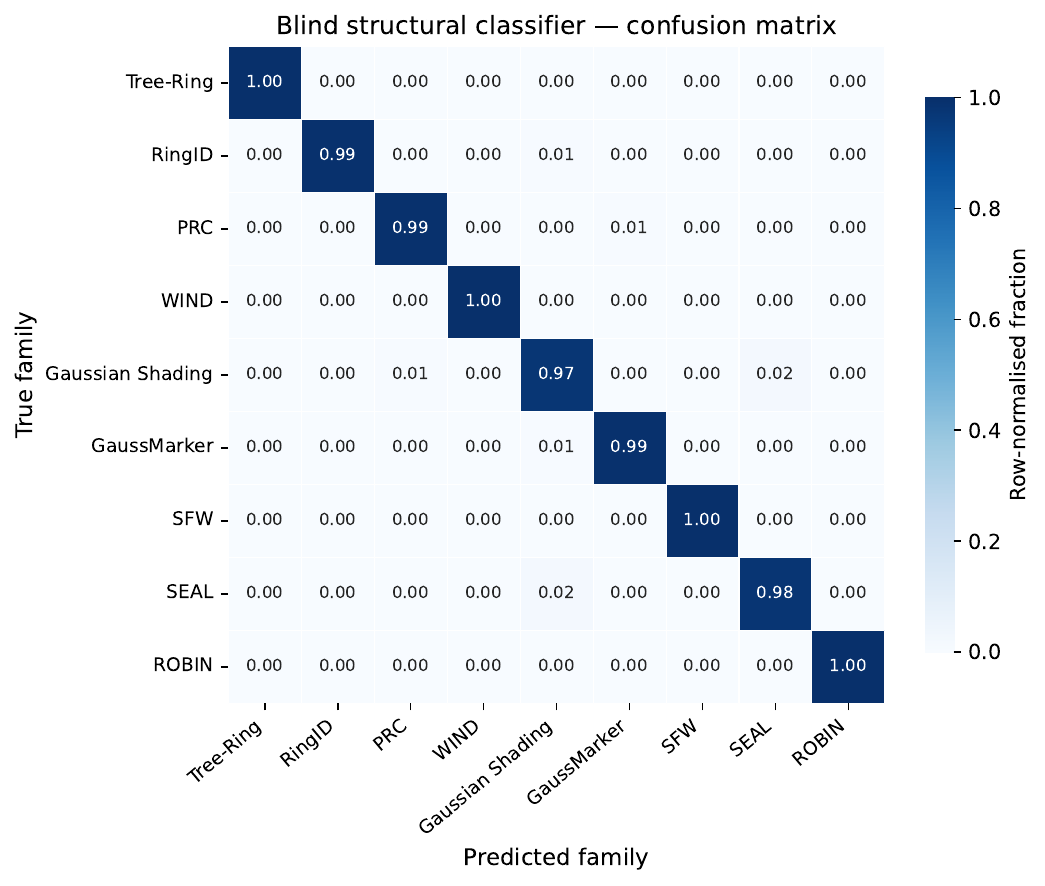}
\caption{Row-normalised confusion matrix of the blind classifier over the nine
families ($\mathbf{z}_T^{\mathrm{inv}}$ features, no keys); off-diagonal mass stays within the
GS/GM/SEAL cluster.}
\label{fig:confusion_blind}
\end{figure}

% -------------------------------------------------------
\subsection{Per-Image Verify-and-Climb Search}
\label{app:adaptive_details}
Table~\ref{tab:adaptive_lambda} compares non-adaptive DRIFT---which fixes a single per-family strength at the $90$th percentile of per-image $\lambda^\star$---against Adaptive DRIFT, which stops at each image's first verifier-rejected rung $\lambda^\star$. The gain concentrates where per-image difficulty is spread (RI, TR, GM, GS) and vanishes where every image needs the same strength (ROBIN).

\begin{table*}[t]
\centering
\small
\setlength{\tabcolsep}{4pt}
\begin{tabular}{l cc cc cc cc cc}
\toprule
& \multicolumn{2}{c}{$\bar{\lambda}\downarrow$} & \multicolumn{2}{c}{ASR (\%)$\uparrow$} & \multicolumn{2}{c}{SSIM$\uparrow$} & \multicolumn{2}{c}{PSNR$\uparrow$} & \multicolumn{2}{c}{LPIPS$\downarrow$} \\
\cmidrule(lr){2-3}\cmidrule(lr){4-5}\cmidrule(lr){6-7}\cmidrule(lr){8-9}\cmidrule(lr){10-11}
Family & Base & Adapt.\ & Base & Adapt.\ & Base & Adapt.\ & Base & Adapt.\ & Base & Adapt.\ \\
\midrule
SEAL  & 0.05 & 0.087 & 98.7  & 100.0 & 0.872 & 0.875 & 29.05 & 28.96 & 0.049 & 0.046 \\
SFW   & 0.10 & 0.074 & 99.7  & 100.0 & 0.736 & 0.787 & 23.76 & 25.74 & 0.099 & 0.075 \\
RI    & 0.20 & 0.161 & 91.0  & 99.7  & 0.658 & 0.755 & 22.45 & 25.32 & 0.179 & 0.112 \\
WIND  & 0.15 & 0.110 & 97.0  & 100.0 & 0.654 & 0.686 & 23.02 & 24.01 & 0.120 & 0.103 \\
ROBIN & 0.15 & 0.150 & 100.0 & 100.0 & 0.723 & 0.719 & 23.98 & 23.85 & 0.124 & 0.123 \\
PRC   & 0.20 & 0.155 & 95.7  & 100.0 & 0.703 & 0.730 & 23.46 & 24.54 & 0.170 & 0.141 \\
TR    & 0.30 & 0.199 & 94.7  & 99.7  & 0.643 & 0.698 & 21.58 & 23.81 & 0.219 & 0.156 \\
GM    & 0.50 & 0.394 & 99.0  & 100.0 & 0.557 & 0.596 & 18.95 & 20.70 & 0.342 & 0.275 \\
GS    & 0.55 & 0.482 & 95.7  & 98.0  & 0.477 & 0.505 & 17.02 & 18.11 & 0.406 & 0.360 \\
\midrule
\textbf{Mean} & 0.244 & \textbf{0.201 } & 96.8 & \textbf{99.7} & 0.669 & \textbf{0.706} & 22.59 & \textbf{23.89} & 0.190 & \textbf{0.155} \\
\bottomrule
\end{tabular}
\caption{\textbf{Non-adaptive vs.\ Adaptive DRIFT} across nine watermark families
($300$ images each; H100). \emph{Base} fixes the $90$th percentile of each
family's per-image $\lambda^\star$; \emph{Adapt.}\
verify-and-climbs to the first verifier-rejected rung using only
black-box accept/reject queries. Best mean per metric in \textbf{bold}.}
\label{tab:adaptive_lambda}
\end{table*}

% -------------------------------------------------------
\subsection{RL Fidelity Refinement: Details and Ablation}
\label{app:rl_details}
Table~\ref{tab:refine} summarises the pre$\to$post fidelity of Stage-III
refinement, Figure~\ref{fig:dppo_curves} shows the DPPO training dynamics, and
Figure~\ref{fig:refine_qual} illustrates fidelity recovery across three
families of decreasing robustness.  The component ablation is retained as
Table~\ref{tab:rl_ablation} in the main paper.

\begin{table}[tb]
\centering
\small
\setlength{\tabcolsep}{5pt}
\renewcommand{\arraystretch}{1.1}
\begin{tabular}{ccc ccc}
\toprule
\multicolumn{3}{c}{Pre-refine (round 0)} & \multicolumn{3}{c}{Post-refine (DPPO)} \\
\cmidrule(lr){1-3}\cmidrule(lr){4-6}
SSIM$\uparrow$ & PSNR$\uparrow$ & LPIPS$\downarrow$
 & SSIM$\uparrow$ & PSNR$\uparrow$ & LPIPS$\downarrow$ \\
\midrule
0.735 & 24.95 & 0.136 & \textbf{0.738} & \textbf{25.19} & \textbf{0.114} \\
\bottomrule
\end{tabular}
\caption{\textbf{Stage-III fidelity refinement} on verifier-rejected attack outputs
($N{=}225$, $25\times9$ families). Empirically, the DPPO controller improves
all three reported mean fidelity metrics over the round-0 attacked image in
$1.16$ rounds on average. Proposition~\ref{prop:feasible} guarantees only that
the retained output remains rejected by the same verifier; the fidelity gains
are empirical.}
\label{tab:refine}
\end{table}

\begin{figure}[!htbp]
  \centering
  \includegraphics[width=0.88\linewidth]{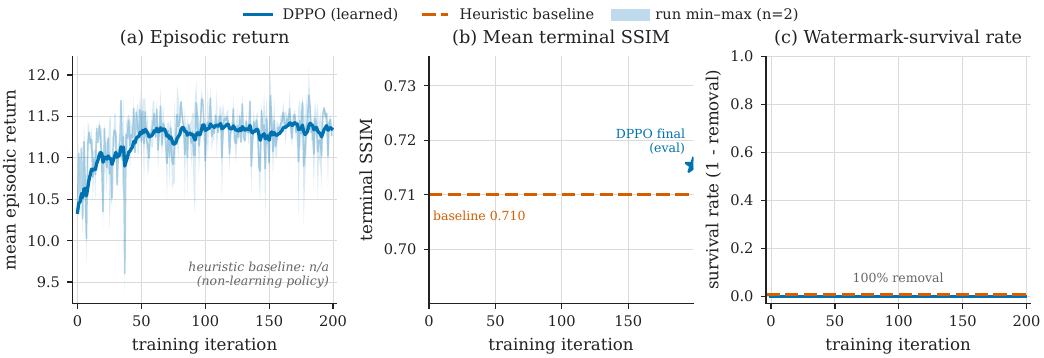}
  \caption{\textbf{DPPO training dynamics.} \textbf{(a)}~Mean episodic return
  rises and plateaus; \textbf{(b)}~mean terminal SSIM improves; \textbf{(c)}~
  returned retained outputs remain verifier-rejected throughout. The first two
  trends are empirical; Proposition~\ref{prop:feasible} explains only the
  same-verifier rejection invariant.}
  \label{fig:dppo_curves}
\end{figure}

\begin{figure}[!htbp]
\centering
\setlength{\tabcolsep}{2pt}
\renewcommand{\arraystretch}{0.6}
\newcommand{\qcell}[1]{\includegraphics[width=0.29\textwidth]{#1}}
\newcommand{\qmet}[3]{\footnotesize PSNR~#1\,/\,SSIM~#2\,/\,LPIPS~#3}
\begin{tabular}{@{}ccc@{}}
\footnotesize\textbf{(a) Watermarked $\mathbf{x}^w$} &
\footnotesize\textbf{(b) Min-$\lambda$ DRIFT} &
\footnotesize\textbf{(c) $+$ DPPO refinement} \\[1pt]
\qcell{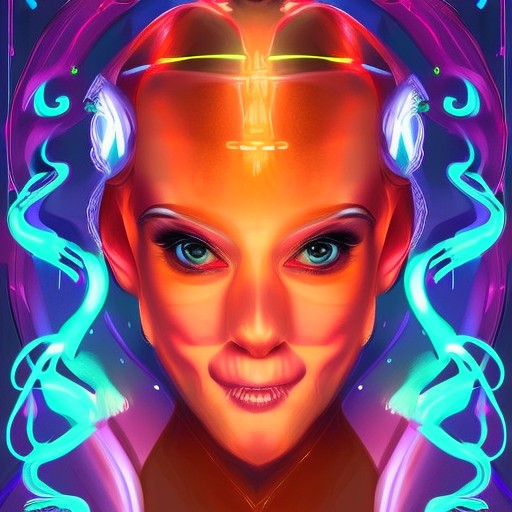} &
\qcell{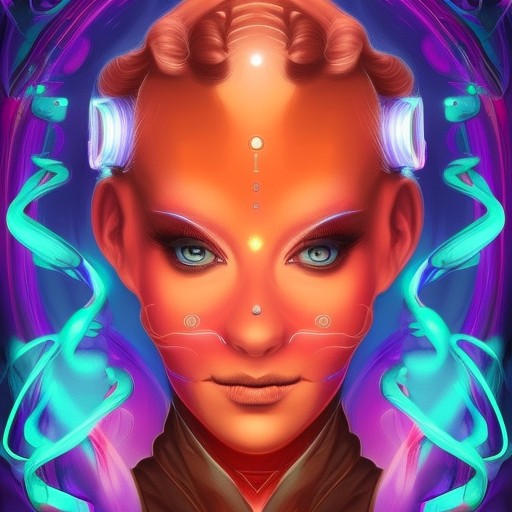} &
\qcell{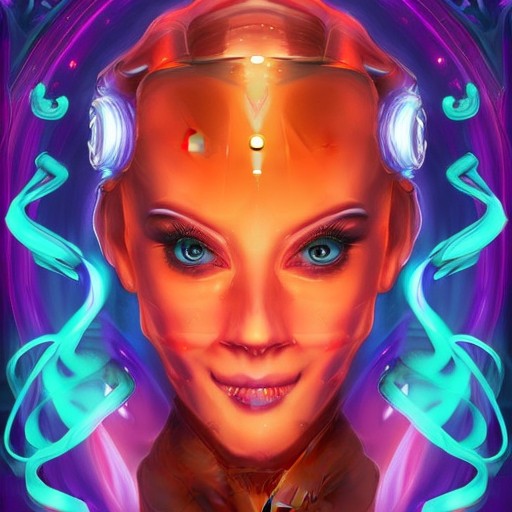} \\[-1pt]
\footnotesize GS ($\lambda\!=\!0.50$), watermark \checkmark\ present &
\qmet{16.3}{.462}{.363}, rejected &
\qmet{\textbf{18.3}}{\textbf{.543}}{\textbf{.194}}, rejected \\[3pt]
\qcell{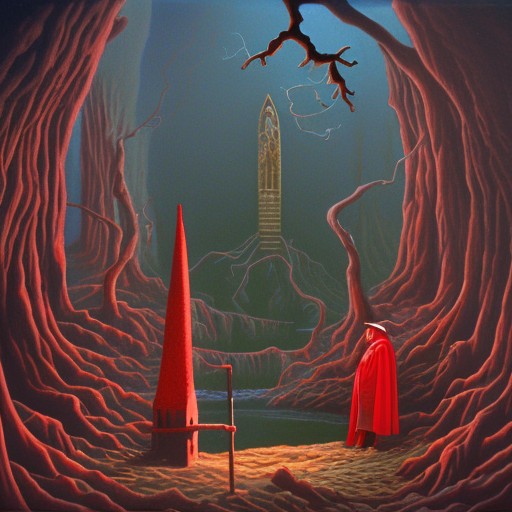} &
\qcell{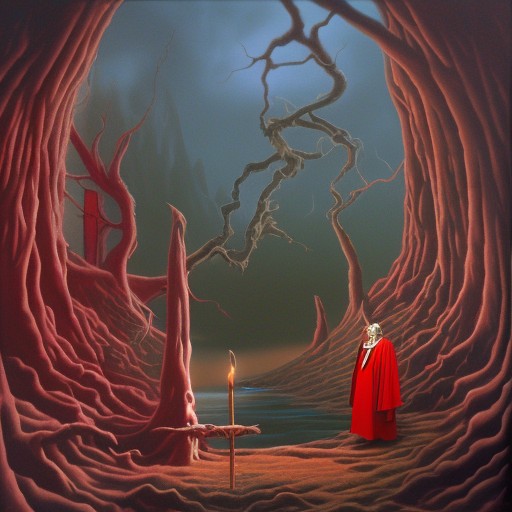} &
\qcell{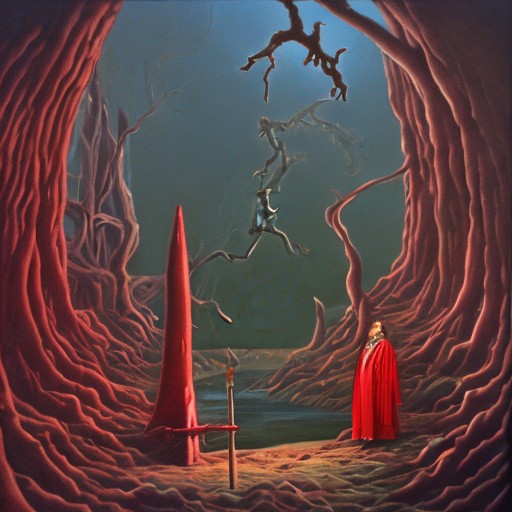} \\[-1pt]
\footnotesize GM ($\lambda\!=\!0.45$), watermark \checkmark\ present &
\qmet{19.5}{.495}{.330}, rejected &
\qmet{\textbf{20.6}}{\textbf{.522}}{\textbf{.179}}, rejected \\[3pt]
\qcell{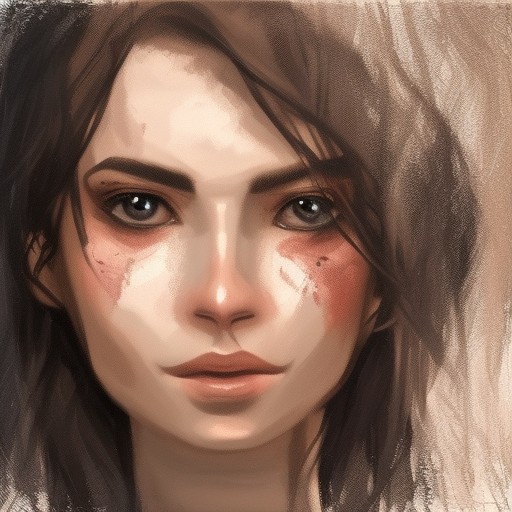} &
\qcell{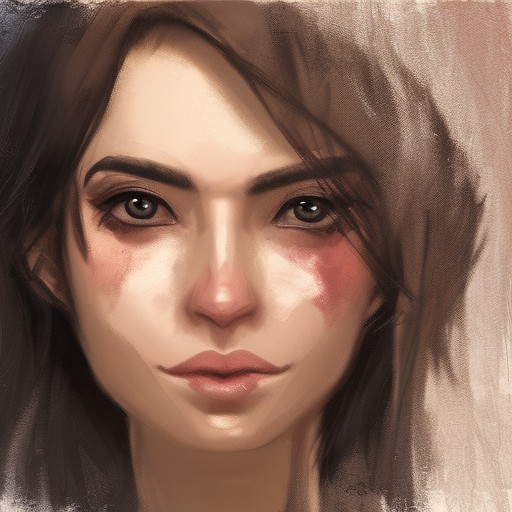} &
\qcell{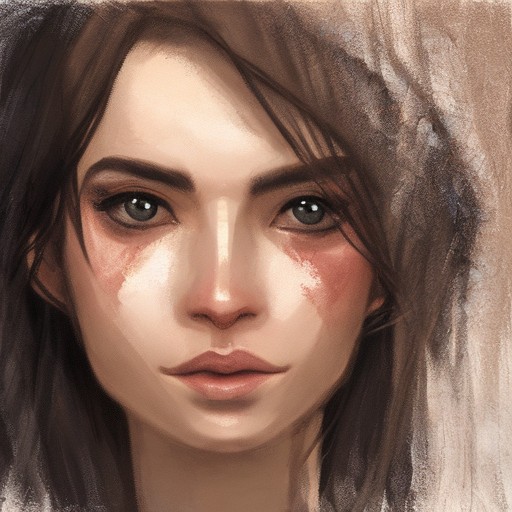} \\[-1pt]
\footnotesize ROBIN ($\lambda\!=\!0.20$), watermark \checkmark\ present &
\qmet{23.3}{.569}{.211}, rejected &
\qmet{\textbf{23.6}}{.558}{\textbf{.142}}, rejected \\
\end{tabular}
\caption{\textbf{RL refinement recovers fidelity at fixed verifier rejection} across three families of decreasing robustness (GS, GM, ROBIN; per-row $\lambda^\star$ in parentheses). \textbf{(a)}~Watermarked image; \textbf{(b)}~Adaptive DRIFT at the first verifier-rejected rung $\lambda^\star$ (the larger strengths needed by GS/GM drift content); \textbf{(c)}~DPPO refinement pulls the result back toward $\mathbf{x}^w$ under a hard verifier-rejection constraint. Per-panel PSNR/SSIM/LPIPS are measured against $\mathbf{x}^w$; better refined value in \textbf{bold}. All panels $512{\times}512$.}
\label{fig:refine_qual}
\end{figure}

% -------------------------------------------------------
\subsection{Qualitative Comparison Across All Watermarking Methods}

\begin{figure*}[t]
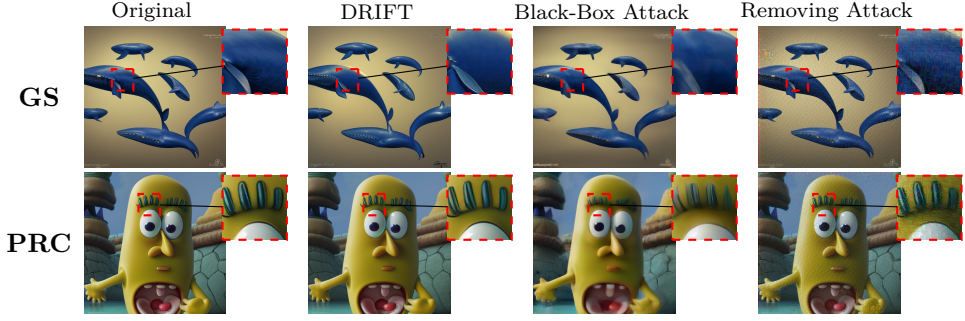

\centering
\renewcommand{\arraystretch}{1.0}
% \methodcol / \imgcol are declared by Figure~\ref{fig:qualitative_zoom_compare}
% in the main text; re-declaring them here would abort with
% "Command \methodcol already defined".
\setlength{\methodcol}{0.055\textwidth}
\setlength{\imgcol}{0.18\textwidth}

\begin{minipage}[c]{\methodcol}\centering\end{minipage}%
\begin{minipage}[c]{\imgcol}\centering \scriptsize Original \end{minipage}%
\begin{minipage}[c]{\imgcol}\centering \scriptsize DRIFT \end{minipage}%
\begin{minipage}[c]{\imgcol}\centering \scriptsize Black-Box Attack \end{minipage}%
\begin{minipage}[c]{\imgcol}\centering \scriptsize Removing Attack \end{minipage}

\vspace{0.3mm}
\begin{minipage}[c]{\methodcol}\centering \small \textbf{GS} \end{minipage}%
\begin{minipage}[c]{\imgcol}\centering \zoomfig{img/shift_advantage_selection/GS/GS_075__1_original.png}{0.72}{0.62}{0.63\linewidth}{0.85cm} \end{minipage}%
\begin{minipage}[c]{\imgcol}\centering \zoomfig{img/shift_advantage_selection/GS/GS_075__2_shift.png}{0.72}{0.62}{0.63\linewidth}{0.85cm} \end{minipage}%
\begin{minipage}[c]{\imgcol}\centering \zoomfig{img/shift_advantage_selection/GS/GS_075__3_blackbox.png}{0.72}{0.62}{0.63\linewidth}{0.85cm} \end{minipage}%
\begin{minipage}[c]{\imgcol}\centering \zoomfig{img/shift_advantage_selection/GS/GS_075__4_remove.png}{0.72}{0.62}{0.63\linewidth}{0.85cm} \end{minipage}

\vspace{0.3mm}
\begin{minipage}[c]{\methodcol}\centering \small \textbf{PRC} \end{minipage}%
\begin{minipage}[c]{\imgcol}\centering \zoomfig{img/shift_advantage_selection/PRC/PRC_071__1_original.png}{2.0}{0.77}{0.63\linewidth}{0.85cm} \end{minipage}%
\begin{minipage}[c]{\imgcol}\centering \zoomfig{img/shift_advantage_selection/PRC/PRC_071__2_shift.png}{2.0}{0.77}{0.63\linewidth}{0.85cm} \end{minipage}%
\begin{minipage}[c]{\imgcol}\centering \zoomfig{img/shift_advantage_selection/PRC/PRC_071__3_blackbox.png}{2.0}{0.77}{0.63\linewidth}{0.85cm} \end{minipage}%
\begin{minipage}[c]{\imgcol}\centering \zoomfig{img/shift_advantage_selection/PRC/PRC_071__4_remove.png}{2.0}{0.77}{0.63\linewidth}{0.85cm} \end{minipage}

\caption{\textbf{Qualitative comparison with baseline attacks on GS and PRC}
(Original / DRIFT / Black-Box / Removing), with zoomed insets. Figure~\ref{fig:qualitative_zoom_compare}
of the main text shows the same comparison on GaussMarker.}
\label{fig:qualitative_zoom_appendix}
\end{figure*}

\label{app:qualitative_all}

Figure~\ref{fig:qualitative_all_methods_a} and Figure~\ref{fig:qualitative_all_methods_b} compare the full Adaptive DRIFT pipeline---the first verifier-rejected rung $\lambda^\star$ followed by DPPO fidelity refinement---across all nine evaluated watermarking methods. For each method, the left column shows the original watermarked image and the right column shows the corresponding attacked-and-refined image. Across the three watermark paradigms, the outputs preserve subject identity, color palette, and composition while remaining rejected by the evaluated verifier. These examples illustrate detector evasion with limited visible artifacts; quantitative fidelity is reported separately.

% --- part (a) ---
\begin{figure*}[t]
    \centering

    \methodfig{GM}{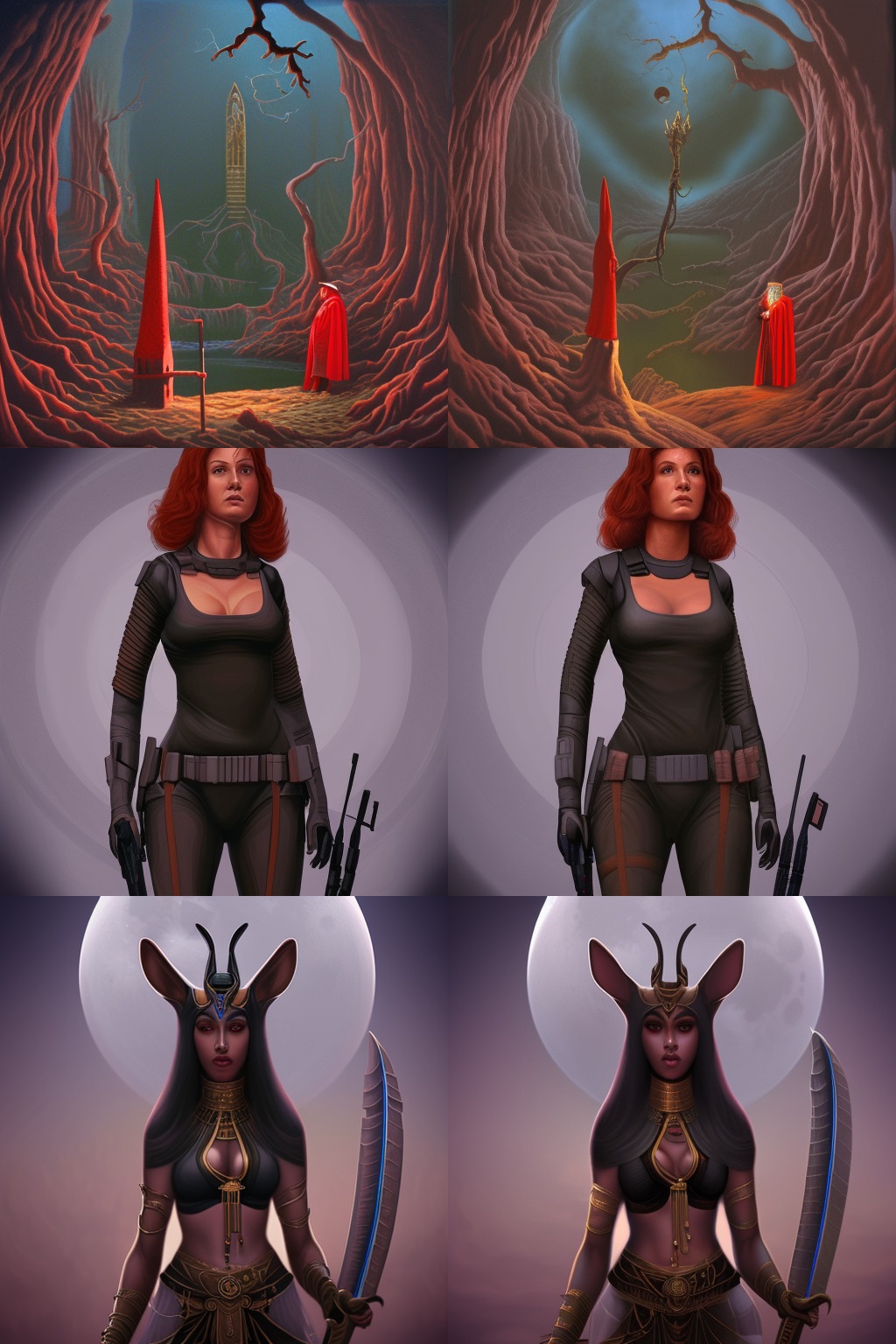}
    \hspace{\colgap}
    \methodfig{SEAL}{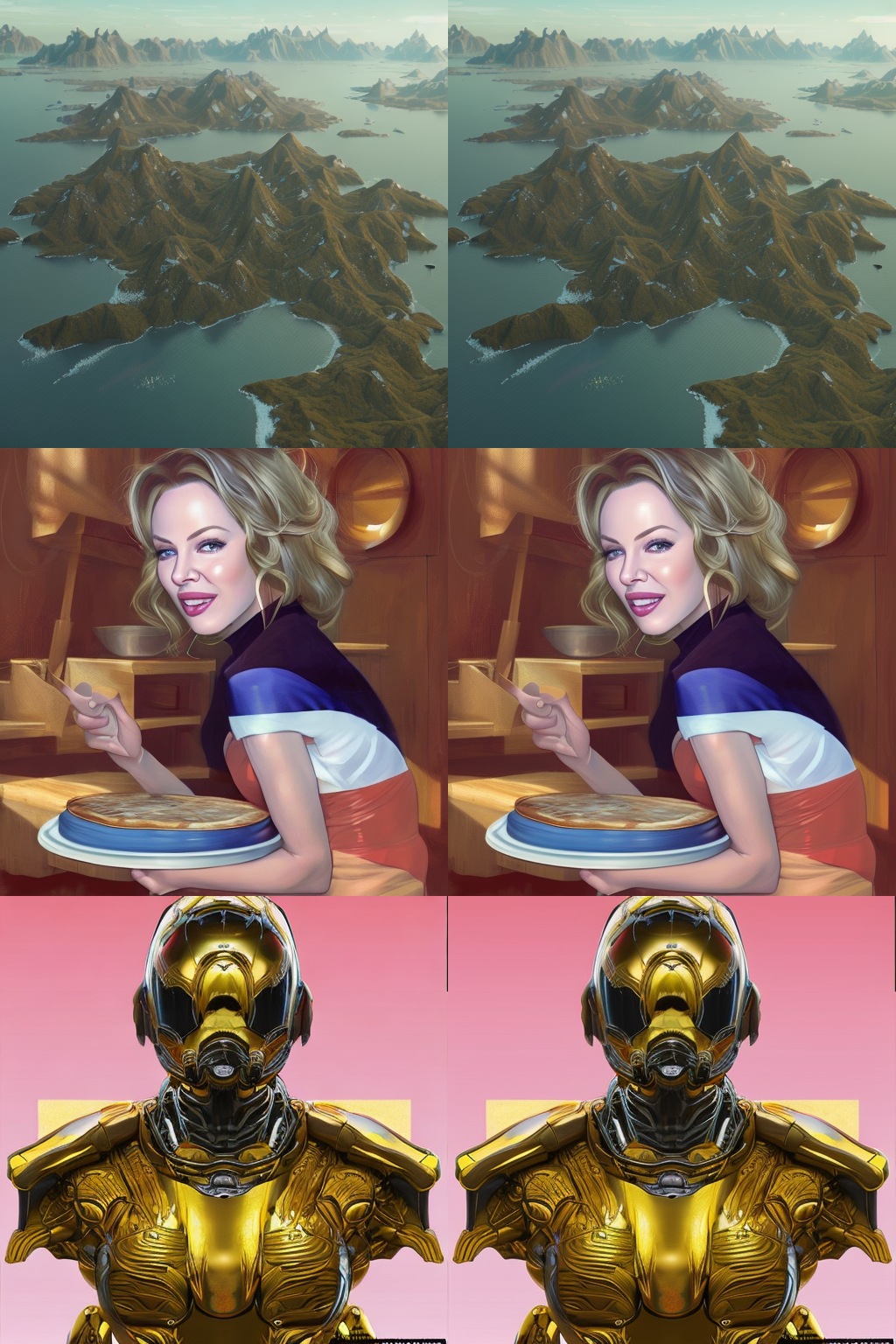}
    \hspace{\colgap}
    \methodfig{PRC}{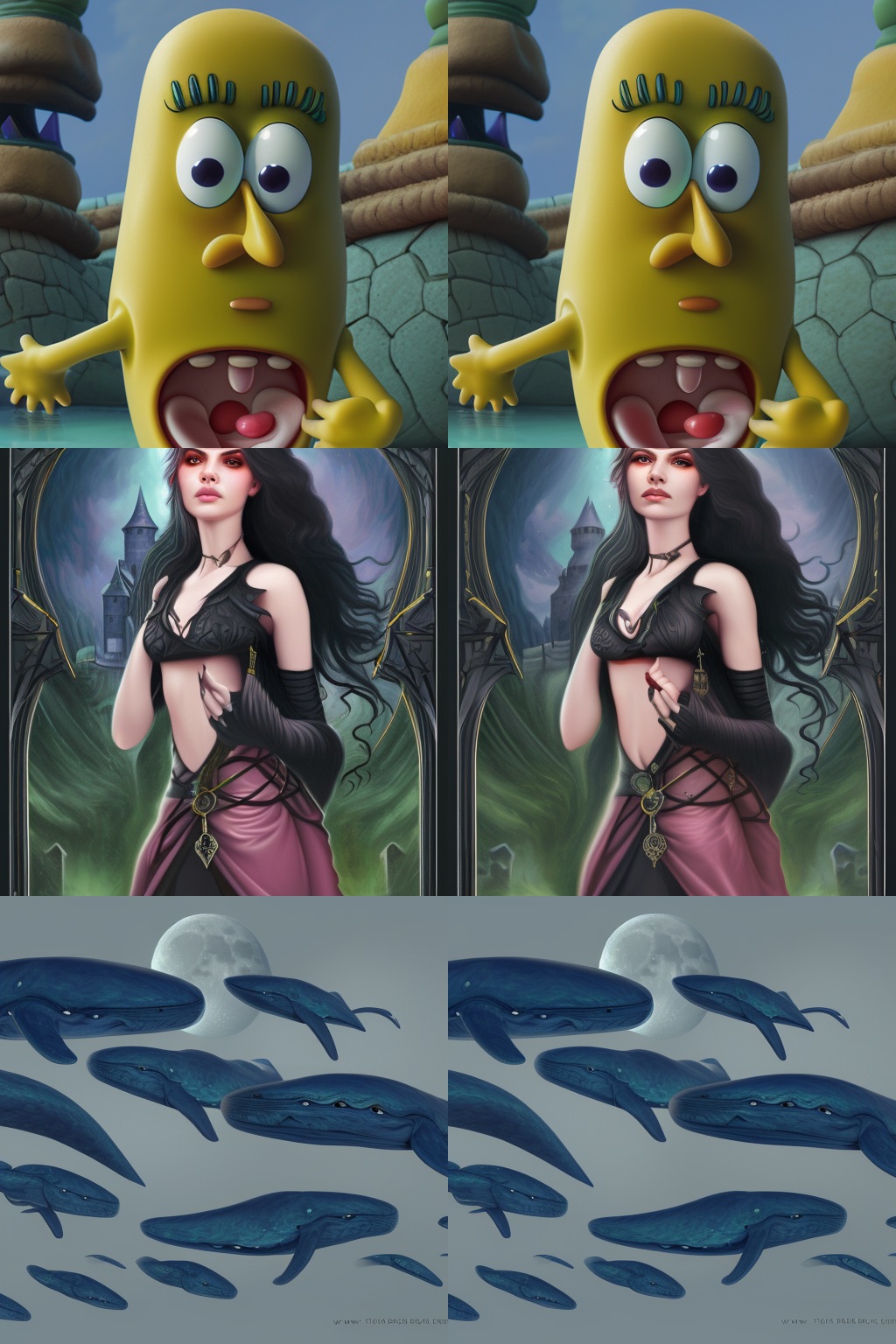}

    \vspace{\rowgap}

    \methodfig{ROBIN}{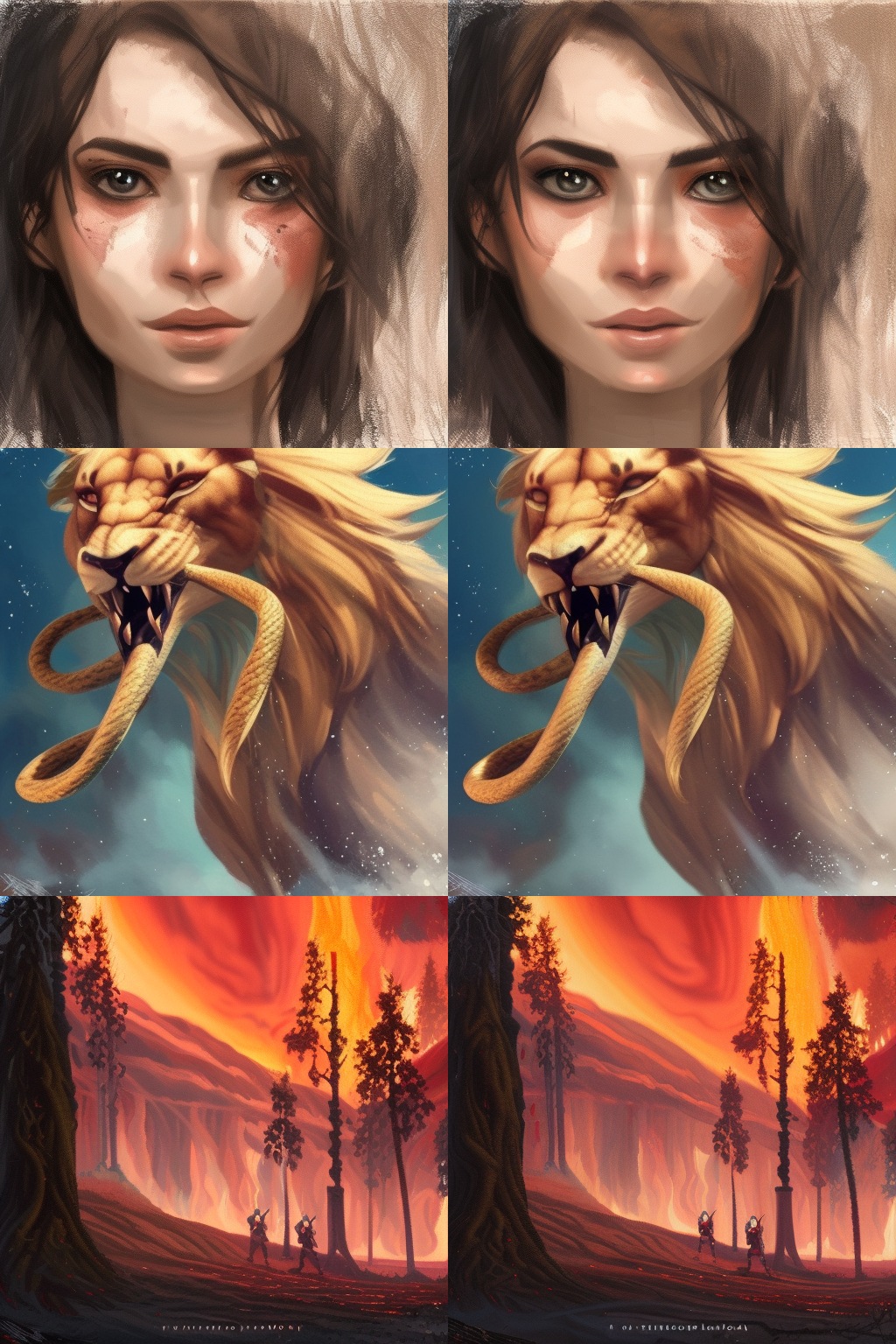}
    \hspace{\colgap}
    \methodfig{RI}{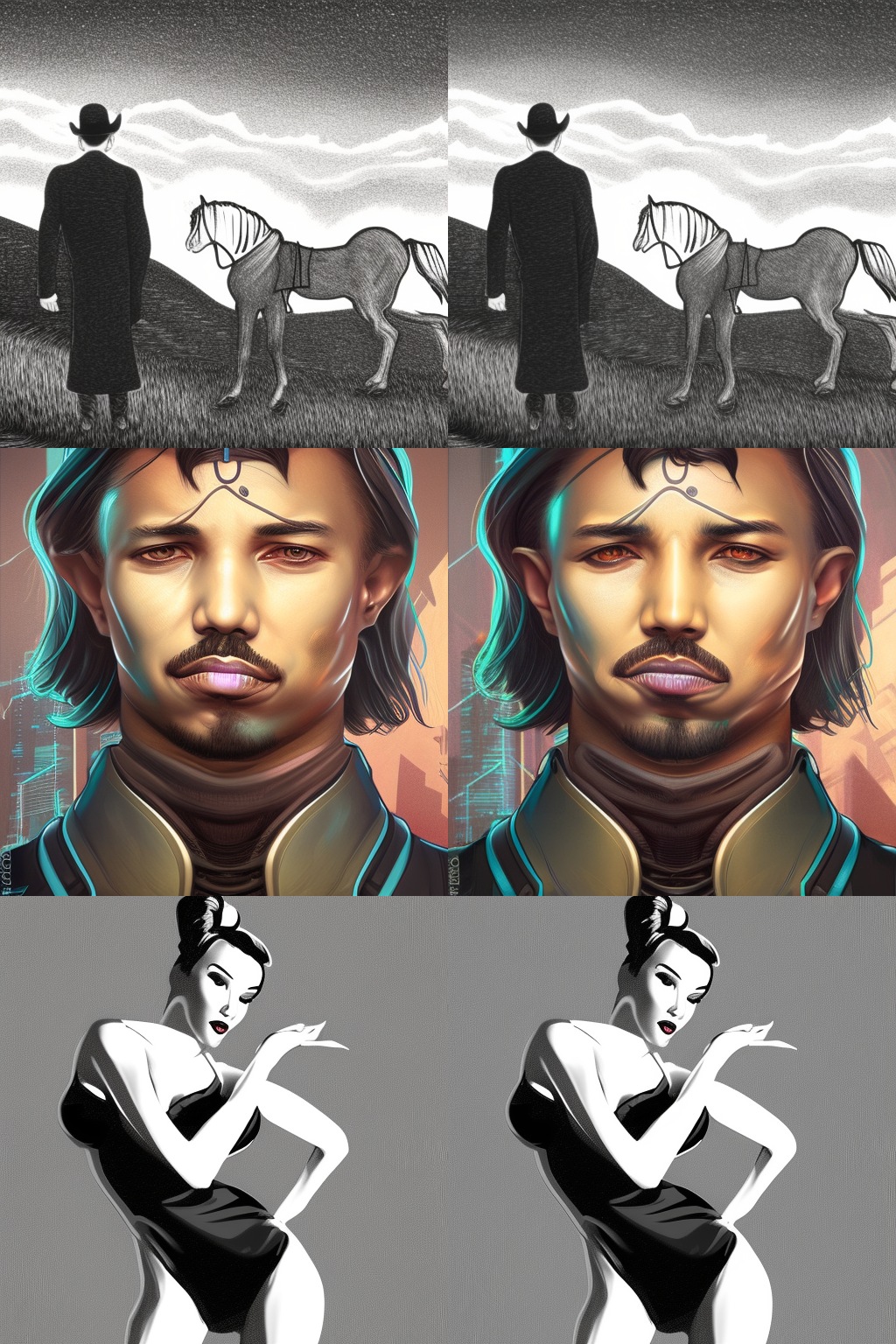}
    \hspace{\colgap}
    \methodfig{TR}{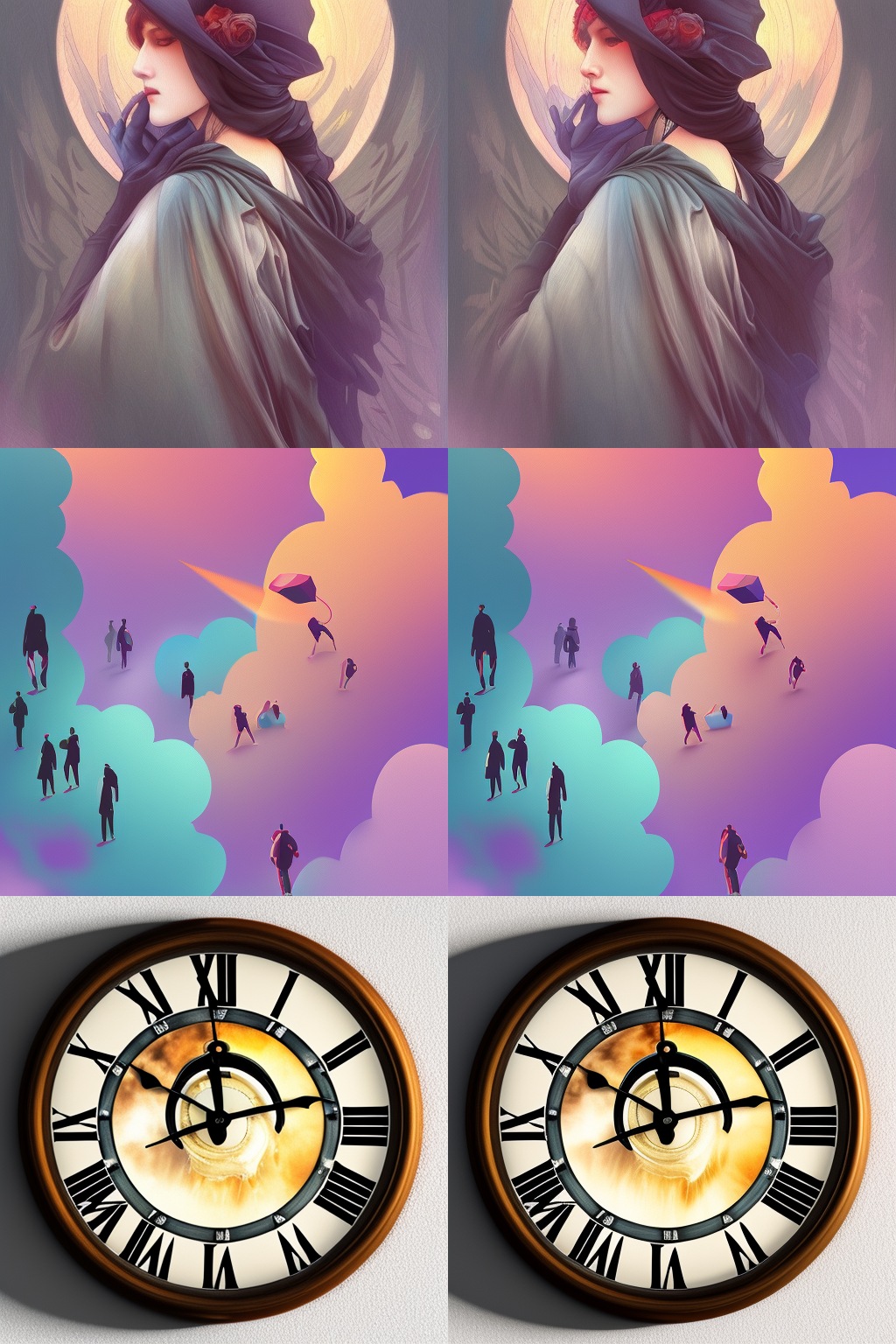}

    \caption{\textbf{(a) Adaptive DRIFT on six watermarking methods.}
    Left: watermarked originals. Right: Adaptive DRIFT outputs (per-image
    first-rejected $\lambda^\star$ + DPPO refinement).}
    \label{fig:qualitative_all_methods_a}
\end{figure*}

% --- part (b) ---
\begin{figure*}[t]
    \centering

    \methodfig{SFW}{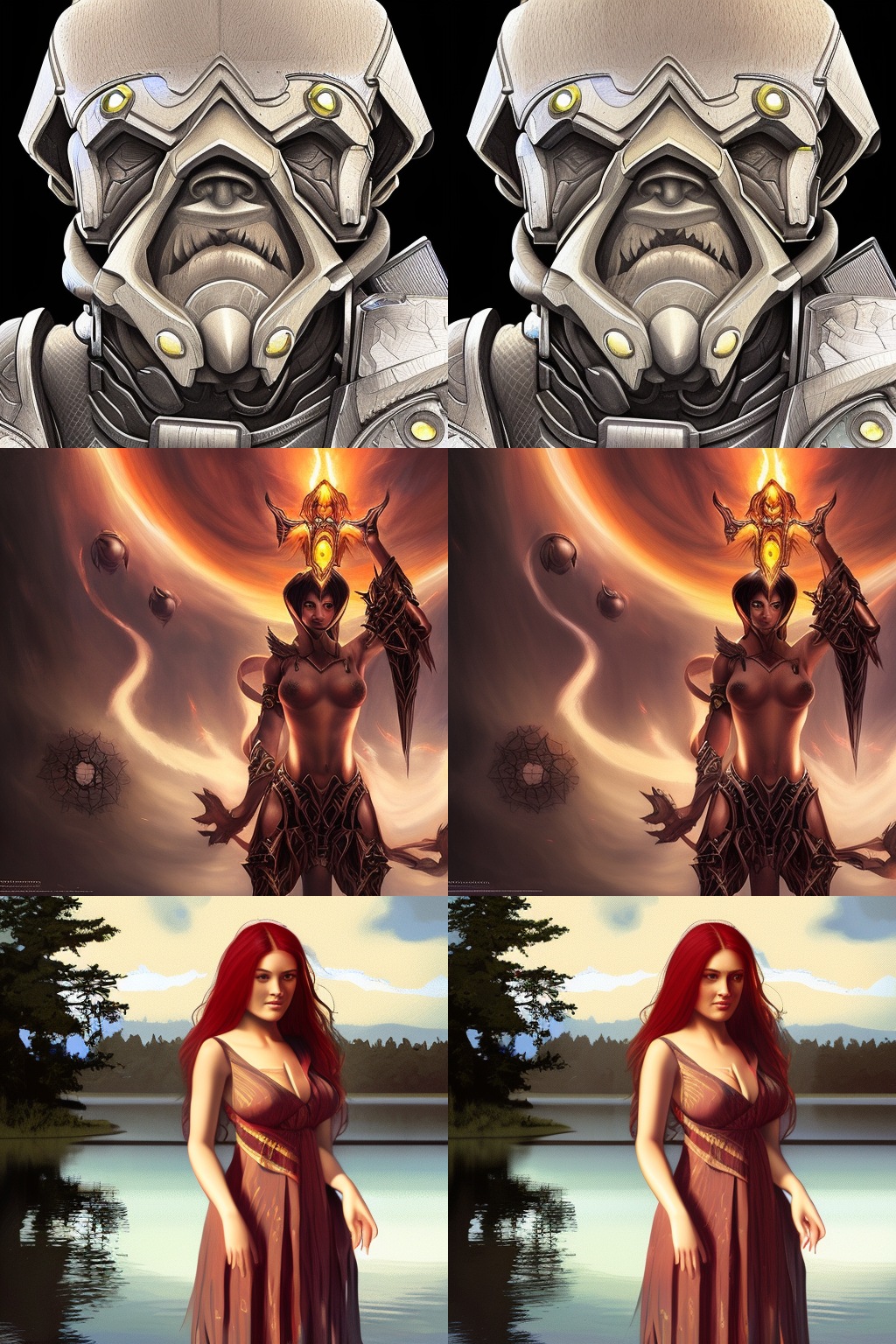}
    \hspace{\colgap}
    \methodfig{GS}{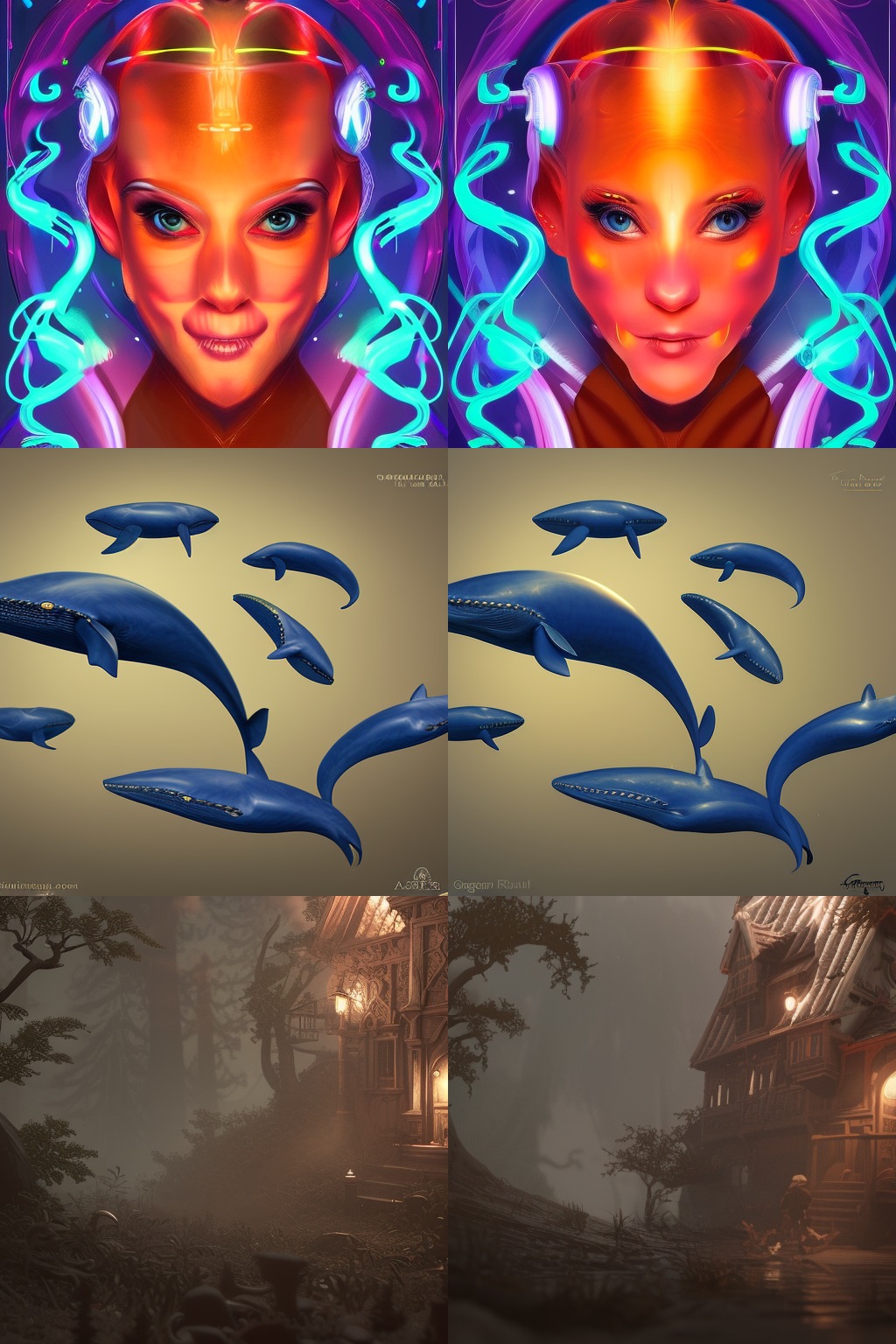}
    \hspace{\colgap}
    \methodfig{WIND}{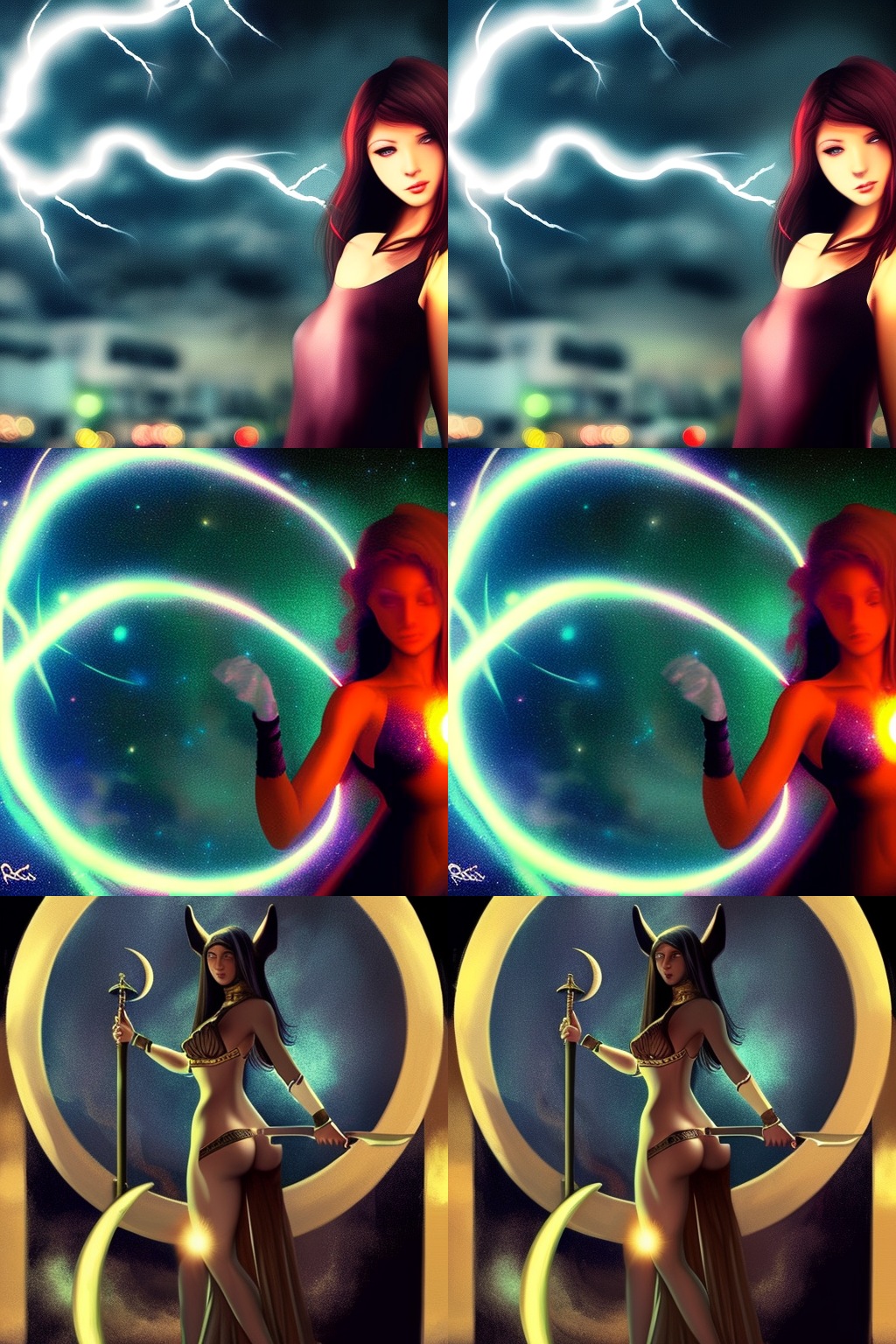}

    \caption{\textbf{(b) Adaptive DRIFT on the remaining three methods.}
    Left: watermarked originals. Right: Adaptive DRIFT outputs (per-image
    first-rejected $\lambda^\star$ + DPPO refinement).}
    \label{fig:qualitative_all_methods_b}
\end{figure*}

% -------------------------------------------------------
\subsection{Effect of Attack Strength on Image Fidelity}
\label{app:robustness_levels}

Figure~\ref{fig:robustness_compare} illustrates the fidelity--evasion trade-off of the \emph{base} (fixed-$\lambda$) DRIFT attack across three representative watermarking methods of increasing robustness as $\lambda$ grows from $0.30$ to $0.60$. For weakly robust methods such as PRC, successful evasion is achieved at $\lambda=0.30$ with virtually no perceptual change relative to the original. For moderately robust methods such as ROBIN, minor variations in fine-grained texture appear at $\lambda=0.45$ but remain within an acceptable perceptual range. For the most robust method Tree-Ring, stronger perturbation at $\lambda=0.60$ is required, introducing slight deviation in high-frequency details while the overall scene structure and semantic content are well preserved. The observed first successful rung varies with the scheme, image, verifier, sampler, search grid, and stochastic realization, motivating Adaptive DRIFT's per-image minimal-strength search.

\begin{figure*}[t]
    \centering
    \setlength{\tabcolsep}{2pt}
    \renewcommand{\arraystretch}{1.0}

    \begin{tabular}{c c c c c}
        \toprule
        \textbf{Method} & \textbf{Original} &
        \textbf{$\lambda=0.30$} &
        \textbf{$\lambda=0.45$} &
        \textbf{$\lambda=0.60$} \\
        \midrule

        \textbf{PRC (Weak)} &
        \includegraphics[width=0.18\textwidth]
            {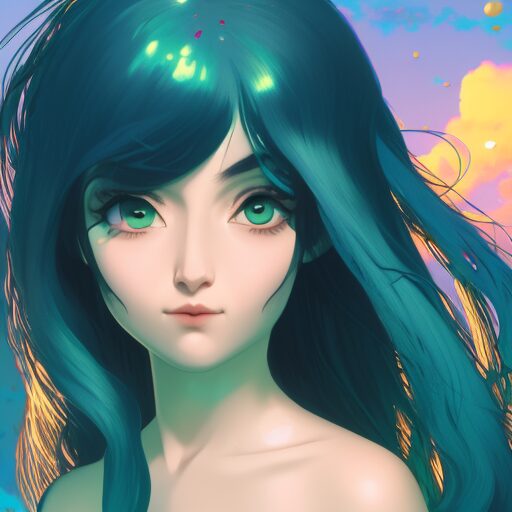} &
        \includegraphics[width=0.18\textwidth]
            {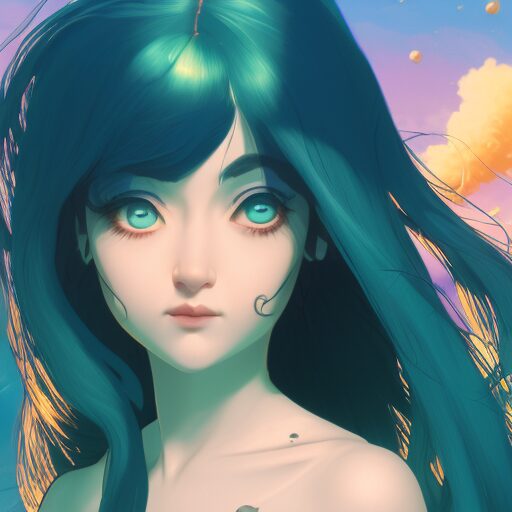} &
        \includegraphics[width=0.18\textwidth]
            {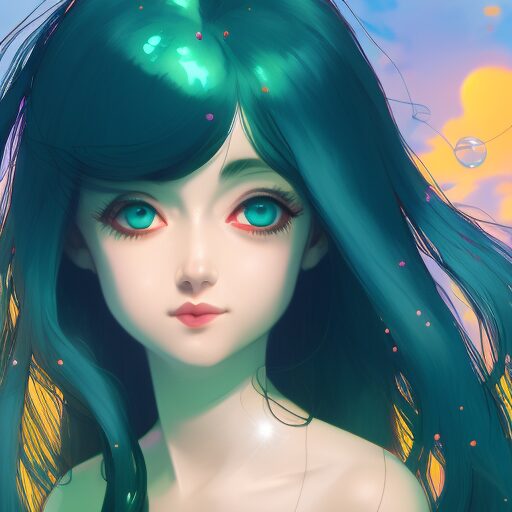} &
        \includegraphics[width=0.18\textwidth]
            {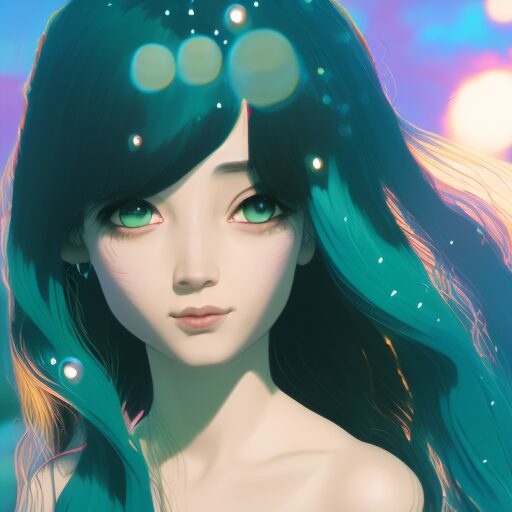} \\[4pt]

        \textbf{ROBIN (Moderate)} &
        \includegraphics[width=0.18\textwidth]
            {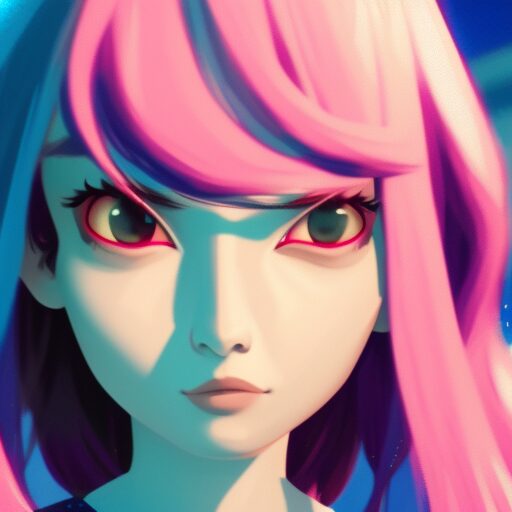} &
        \includegraphics[width=0.18\textwidth]
            {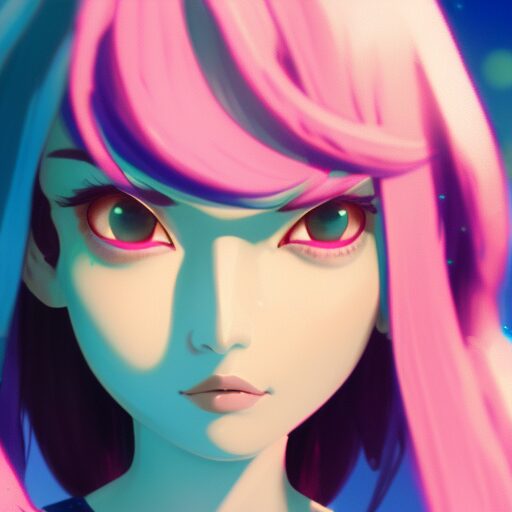} &
        \includegraphics[width=0.18\textwidth]
            {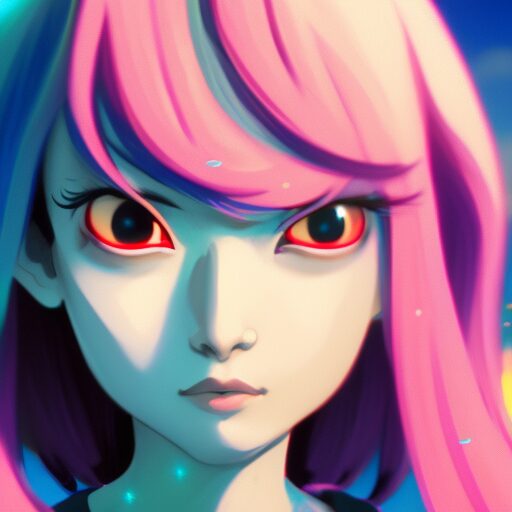} &
        \includegraphics[width=0.18\textwidth]
            {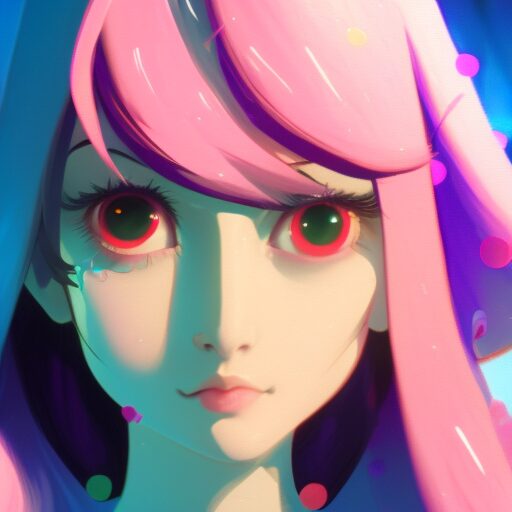} \\[4pt]

        \textbf{Tree-Ring (Strong)} &
        \includegraphics[width=0.18\textwidth]
            {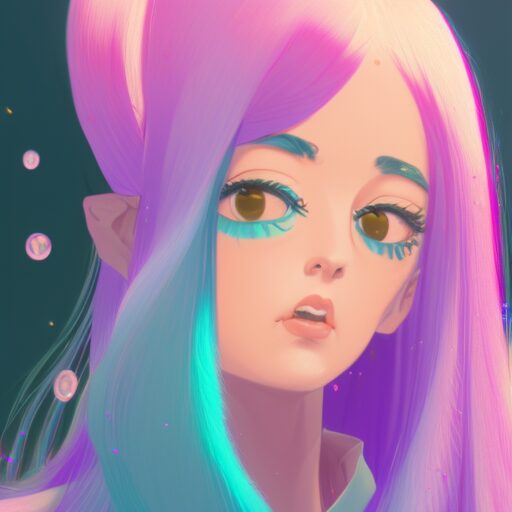} &
        \includegraphics[width=0.18\textwidth]
            {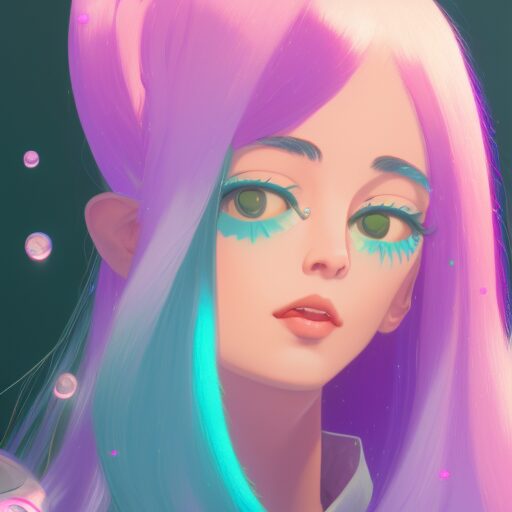} &
        \includegraphics[width=0.18\textwidth]
            {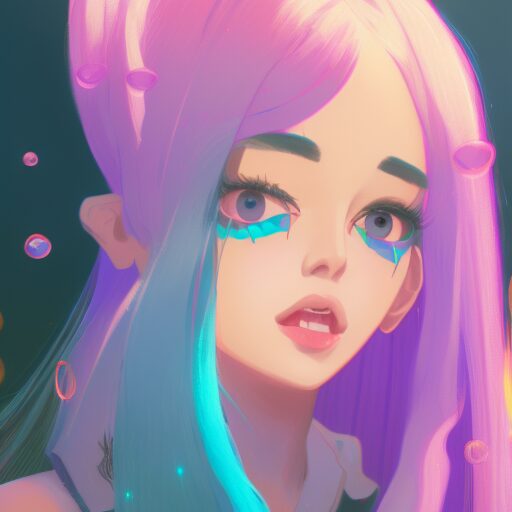} &
        \includegraphics[width=0.18\textwidth]
            {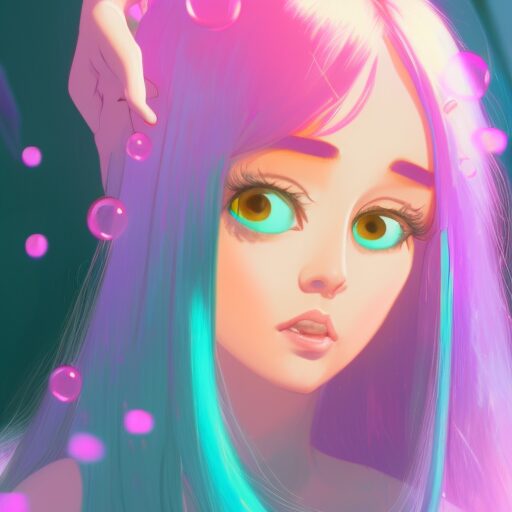} \\

        \bottomrule
    \end{tabular}

    \caption{\textbf{Fidelity--evasion trade-off of base DRIFT under
    different attack strengths.} PRC, ROBIN, and Tree-Ring represent weakly,
    moderately, and highly robust methods. As $\lambda$ increases, the attacked
    outputs reveal clear differences in the trade-off across robustness levels,
    motivating the per-image minimal-strength search of Adaptive DRIFT.}
    \label{fig:robustness_compare}
\end{figure*}

% -------------------------------------------------------
\subsection{Trajectory Decoupling: Noise Distance vs.\ Attack Strength}
\label{app:noise_distance}

Figure~\ref{fig:noise_distance} plots the mean $L_1$ and $L_2$ distances between the DDIM-inverted noise of the attacked image and the reference noise $\boldsymbol{\epsilon}^w$ as a function of $\lambda$. Both metrics increase monotonically across the measured grid, providing empirical evidence that stronger attacks progressively decouple the recovered latent from its original trajectory. At $\lambda=0.70$, the curves lie in the narrow bands $L_1\approx1.05$--$1.07$ and $L_2\approx1.31$--$1.34$. These values approach the random-Gaussian diagnostics $L_1=2/\sqrt{\pi}$ and per-coordinate RMS $L_2=\sqrt{2}$. Corollary~\ref{cor:shift_2d}, under Assumption~\ref{asm:shift_prior}, formalizes only the corresponding squared-$L_2$ reference $\mathbb{E}[\|\cdot\|_2^2]=2d$ (up to its stated perturbation bound); the $L_1$ value additionally requires independent Gaussian coordinates. For non-noise-space schemes, both values are used only as diagnostic comparisons.

\begin{figure*}[t]
    \centering
    \includegraphics[width=\textwidth]{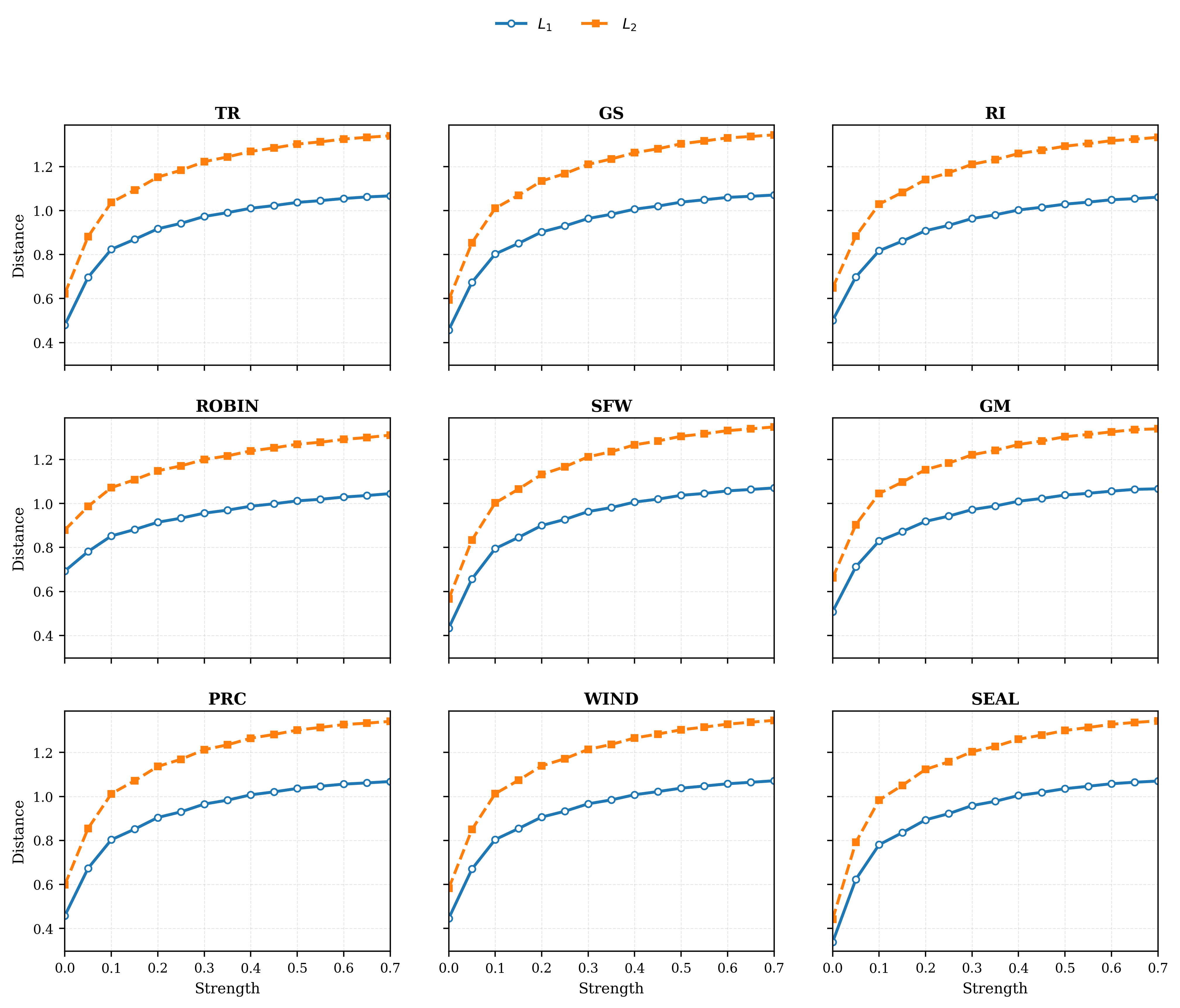}
    \caption{\textbf{Mean $L_1$ and $L_2$ noise distances vs.\ attack strength
    $\lambda$ across nine watermarking methods.} Across the measured grid, both
    metrics increase and approach the random-Gaussian diagnostic values. This
    trend is empirical: Theorem~\ref{prop:noise_distance_formal} bounds
    fixed-depth source dependence and does not predict monotonicity of these
    distances. Corollary~\ref{cor:shift_2d} conditionally formalizes only the
    squared-$L_2$ $2d$ reference; the $L_1$ value additionally assumes
    independent Gaussian coordinates.}
    \label{fig:noise_distance}
\end{figure*}

\subsection{Strength-Sensitivity Result}
\label{app:lambda_sensitivity}

\begin{figure*}[t]
\centering
\includegraphics[width=0.82\textwidth]{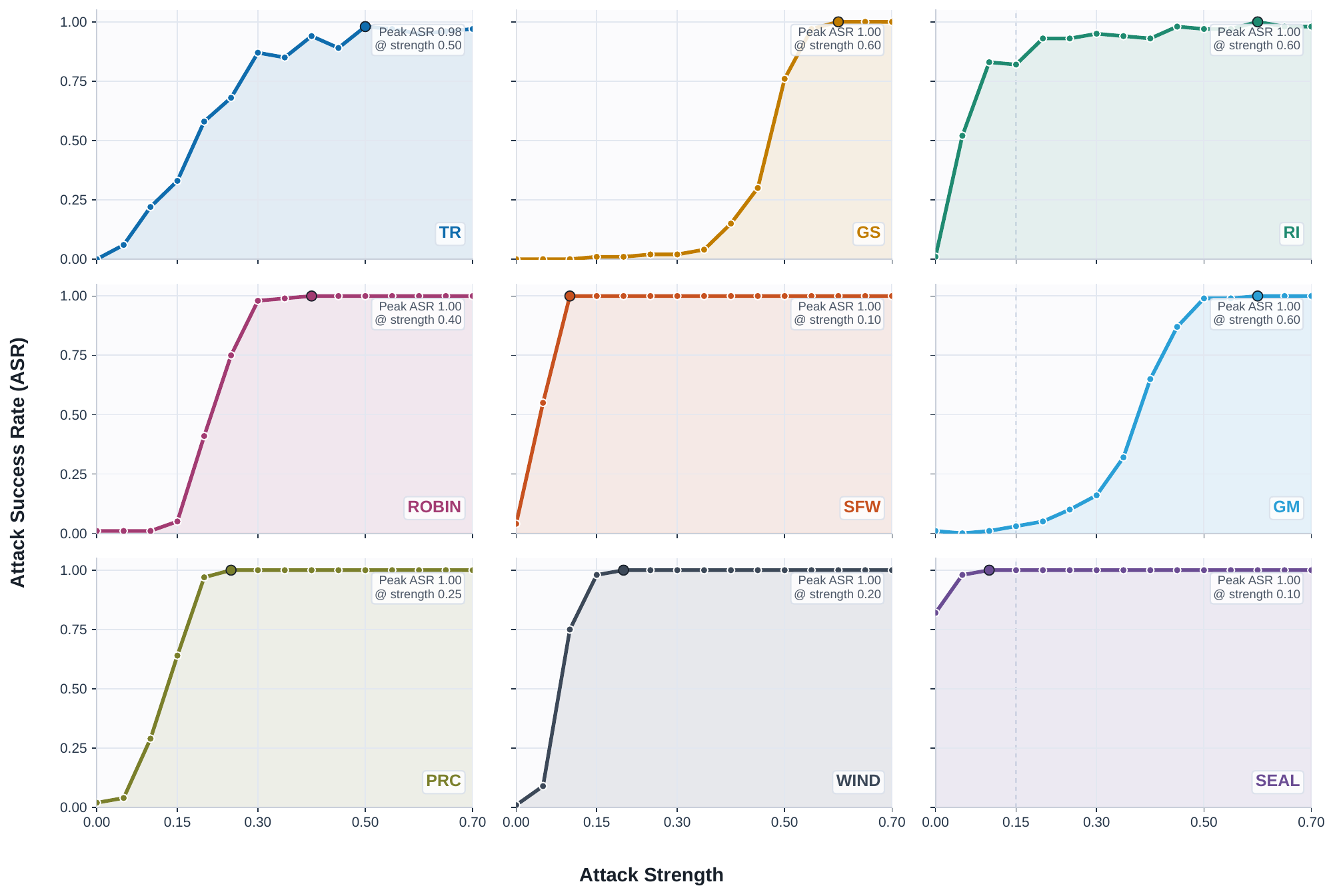}
\caption{\textbf{Sensitivity to the re-noising strength $\lambda$.} ASR as a
function of $\lambda$ for all nine watermark families. The widely separated
saturation thresholds motivate selecting the smallest successful strength for
each image rather than applying one global value.}
\label{fig:lambda_sensitivity}
\end{figure*}

Figure~\ref{fig:lambda_sensitivity} exposes why a single global strength is inefficient. SFW  and SEAL saturate at $\lambda=0.10$ , WIND and PRC at $0.20$--$0.25$, and ROBIN at $0.40$, whereas Tree-Ring, Gaussian Shading, RingID, and GaussMarker require $0.50$--$0.60$. Nevertheless, the peak ASR remains $98$--$100\%$ across all nine families. The dominant variation is therefore the minimum effective strength, precisely the quantity targeted by verify-and-climb.